\documentclass[11pt]{article}
\usepackage[margin=1in]{geometry}
\usepackage{mathpazo}
\usepackage[colorlinks,citecolor=blue,linkcolor=blue,urlcolor=blue]{hyperref}
\usepackage[backend=biber,style=alphabetic,natbib=true,maxbibnames=99]{biblatex}
\usepackage[utf8]{inputenc}
\usepackage{xcolor}
\definecolor{mydarkblue}{rgb}{0,0.08,0.85}
\usepackage{amsfonts}
\usepackage{caption}
\usepackage{subcaption}
\usepackage{amsmath}
\usepackage{amsfonts}
\usepackage{amssymb}
\usepackage{amsbsy}
\usepackage{amsthm}
\usepackage{xcolor}
\usepackage{algorithm,algpseudocode}

\newcommand{\dw}{\calW}
\newcommand{\dk}{\mathrm{d}_{\mathrm{K}}}

\usepackage{selfSelection}

\renewcommand{\hat}[1]{\widehat{#1}}
\renewcommand{\bar}[1]{\overline{#1}}

\title{Estimation of Standard Auction Models}
\author{
\begin{tabular}{cc}
& \\
\textbf{Yeshwanth Cherapanamjeri} & \textbf{Constantinos Daskalakis}\\
\small{University of California Berkeley} & \small{Massachusetts Institute of Technology}\\
\small{\texttt{yeshwanth@berkeley.edu }} & \small{\texttt{costis@csail.mit.edu}}\\
& \\
\textbf{Andrew Ilyas} & \textbf{Manolis Zampetakis}\\
\small{Massachusetts Institute of Technology} & \small{University of California Berkeley}\\
\small{\texttt{ailyas@mit.edu}} & \small{\texttt{mzampet@berkeley.edu}}\\
& \\
\end{tabular}
}
\date{}

\begin{document}
\maketitle

\begin{abstract}
  We provide efficient estimation methods for first- and second-price auctions
under independent (asymmetric) private values and partial observability. Given a
finite set of observations, each comprising the identity of the winner and the
price they paid  in a sequence of identical auctions, we provide algorithms for
non-parametrically estimating the bid distribution of each bidder, as well as
their value distributions under equilibrium assumptions.  We provide
finite-sample estimation bounds which are uniform in that their error rates do
not depend on the bid/value distributions being estimated.  Our estimation
guarantees advance a body of work in Econometrics wherein only identification
results have been obtained
(e.g.~\cite{athey2002identification,athey2007nonparametric}), unless the setting
is symmetric~(e.g.~\cite{morganti2011estimating,menzel2013large}), parametric
(e.g.~\cite{athey2011comparing}), or all  bids are observable
(e.g.~\cite{guerre2000optimal}).  Our guarantees also provide computationally and
statistically effective alternatives to classical techniques from reliability
theory~\citep{meilijson1981estimation}. Finally, our results are immediately
applicable to Dutch and English auctions.

\end{abstract}

\section{Introduction}
Estimating value and/or bid distributions from an observed sequence of auctions
is a fundamental challenge in Econometrics with direct practical applications.
For example, these fundamentals allow one to analyze the performance 
of an auction and make counterfactual predictions about alternatives.
The difficulty of this problem depends on the format of the
 auctions and the structure of the observed information
from each one, as well as  how the fundamentals of bidders are interrelated
and vary across the sequence of observations. 

In this paper, we study a basic version of the afore-described estimation
challenge, wherein the auction format and the bidder distributions stay fixed
across  
observations, and the bidders have
independent private values (which are independently resampled across different
observations). 
The auction formats that we consider are first- and second-price
auctions, as well as Dutch and English auctions. What will make our problem
challenging is that (i) our bidders are {\em ex ante} asymmetric, drawing their
independent private values from different distributions; (ii) we
will make no parametric assumptions about these distributions; and (iii) we will
only be observing the identity of the winner and the price they paid but not the
losing bids. Under this observational model and our independent private
values assumption above, we can focus our attention on first- and second-price
auctions, and our results automatically extend to Dutch and English auctions. 

In the above settings, we give computationally and sample efficient methods
for estimating all agents' bid distributions and (under equilibrium
assumptions) value distributions: 

\begin{itemize}
    \item[$\blacktriangleright$] In the case of first-price auctions, we provide
    finite-sample estimation guarantees under L\'evy, Kolmogorov and Total
    Variation distance with minimal assumptions.  Under (a condition weaker
    than) a lower bound on the density of the bid distributions (although we
    actually do not need existence of densities), 
    Theorem~\ref{thm:1stPrice:fullSupport} shows that the bid distributions can
    be estimated to within~$\varepsilon$ in L\'evy distance, using
    $1/\varepsilon^{O(k)}$ samples, where $k$ is the number of bidders.
    Theorem~\ref{thm:main_lb} shows that the exponential dependence on $k$ is
    necessary, and Theorem~\ref{thm:lb_kol} shows that L\'evy distance cannot be
    strengthened to Kolmogorov distance. Sidestepping the exponential sample
    dependence on $k$, strengthening the estimation distance, and removing the
    density lower bound assumption, Theorem~\ref{thm:1stPrice:likely} shows
    that,  assuming only  continuity  of their cumulative functions, the bid
    distributions can be estimated to within $\varepsilon$ in Kolmogorov
    distance on their {\em effective supports}, i.e.~the part of the support
    that is likely to be observed, defined in~\citet{blum2015learning}. Finally,
    under Lipschitzness assumptions on the densities of the bid distributions,
    Theorem~\ref{thm:1stPrice:density} improves the latter to
    $\varepsilon$-error in Total Variation distance. Our sample requirements for
    estimation over the effective supports of the bid distributions under either
    Kolmogorov or TV distance are dramatically improved to {\em logarithmic} in
    the number of bidders and benign in $1/\varepsilon$.  Finally, assuming that
    bidders use Bayesian Nash equilibrium strategies,
    Theorems~\ref{thm:value_estimation_full} and~\ref{thm:value_estimation} show
    that bidders' value distributions can be estimated over their full and,
    respectively, effective supports with similar sample sizes as those needed
    for the estimation of bid distributions.

    All of our estimation algorithms run in polynomial time in their sample
    sizes, and all our estimation error bounds are uniform in that they do not
    depend on the bid/value distributions being estimated, unlike the
    instance-dependent rates that commonly arise from the use of kernel density
    estimation methods. It is also important to note that we estimate the value
    distributions in L\'evy distance (in fact, in the stronger
    notion of Wasserstein distance) and this is sufficient for the purposes of
    performing counter-factual predictions about the revenue that would result
    from running alternative auctions~\citep{brustle2020multi}.

    \item[$\blacktriangleright$] In the case of second-price auctions,
    Theorem~\ref{thm:sec_price_fin_samp} establishes that bid distributions can be
    estimated to within $\varepsilon$ in Kolmogorov distance over their entire
    supports assuming upper and lower bounds on their density functions. Again the
    sample complexity scales as $1/\varepsilon^{O(k)}$. This result poses major
    technical challenges, requiring a computationally and statistically effective,
    fixed point computation alternative to \citet{meilijson1981estimation}'s method.
    We again sidestep the exponential dependence of the required sample size on $k$,
    by considering estimation over the effective support of the distributions in a
    setting, similar to that proposed by~\cite{blum2015learning} for first-price
    auctions, where we can insert bids to the auction or, equivalently, set a
    reserve price (see discussion in Section~\ref{sec:1stPrice:additionalBidder}).
    In this setting, Theorem~\ref{thm:sp_bid_insert} shows that bid distributions
    can be estimated to within $\varepsilon$ in Kolmogorov distance over their
    effective supports, using a sample size that is polynomial in both
    $1/\varepsilon$ and $k$. Similar to Theorem~\ref{thm:1stPrice:density}
    estimation in Kolmogorov distance can be turned to estimation in Total Variation
    distance under Lipschitzness of the densities. Of course, assuming that the
    bidders bid according to the truthful bidding equilibrium, our estimation
    results for bid distributions automatically translate to estimation results for
    value distributions.
\end{itemize}

To the best of our knowledge, our results are the first finite-sample estimation
guarantees for the general problem we consider. In particular: 
\begin{itemize}
    \item There is an extensive line of work on identification of bid and value
    distributions from complete or partial observations of bids;
    see~\cite{athey2007nonparametric} for a survey. In our setting of
    independent private values, \cite{athey2002identification} show that with
    infinite samples, bid distributions are identifiable from the distribution
    of the winners' identity and paid price.  
    
    Identification results for bid and value distributions have been
    established in the presence of correlated values, alternative auction
    formats, and unobserved heterogeneity, or unknown numbers of bidders 
    \citep{paarsch1992deciding,
    laffont1993structural, laffont1995econometrics, donald1996identification,
    baldwin1997bidder, donald2003empirical,
    bajari2003deciding, haile2001auctions, luo2020identification,
    haile2003nonparametric, mbakop2017identification, hu2013identification}. In contrast,
    here we focus on the IPV framework (albeit, in its more challenging
    asymmetric case) and standard auction mechanisms (first-price, second-price,
    Dutch, and English auctions), and provide finite-sample {\em estimation}
    results in these settings. 

\item On the estimation front, \citet{morganti2011estimating} and
\citet{menzel2013large} both provide estimators of bid distributions in first-
and second-price auctions (in fact they provide estimation using any order
statistics of the bid distributions) under the restrictive assumption that the
bidders are {\em symmetric}. They obtain rates for estimation of the bid
distributions over the full support, which  degrade exponentially with the
number of bidders. In comparison, we get similar rates for the significantly
more challenging asymmetric setting, and also provide drastically better
estimation rates (with logarithmic dependence on the number of bidders) on the
effective supports of the distributions.

\item In terms of non-parametric estimation of value distributions from bid
distributions in first-price auctions, \cite{guerre2000optimal} provide
estimation algorithms which operate under the restrictive assumption that the
bidders are symmetric. Their estimation makes use of the explicit formula for
the Bayesian Nash equilibrium in the symmetric case.  Later work
\citep{campo2003asymmetry,bajari2003deciding,
Krasnokutskaya2011IdentificationAE, haile2003nonparametric} extends these
results to the more general asymmetric setting, where the Bayesian Nash
equilibrium has no closed-form expression; as such, these analyses are typically
limited to the setting with only two unique bidder types. 
Our algorithms operate in the latter
(significantly more challenging) setting, with the additional challenges of (a)
allowing each bidder to have their own unique value distribution (b)
only observing the winning bid, rather than all agents' bids; (c) not using
higher-order differentiability assumptions used in prior work while 
providing uniform convergence bounds (i.e., not depending on the
distributions being estimated).

\item In the computer science literature, there has been work on non-parametric
estimation of bid distributions in first-price auctions, under the stronger
assumption that the econometrician can insert bids which do not influence the
bidding behavior of the bidders~\cite{blum2015learning}. We compare to that
setting and work in Section~\ref{sec:1stPrice:additionalBidder}, explaining that
our work obtains substantial improvements on their rates. 
\end{itemize}

Besides the non-parametric identification work on auctions, discussed above,
there has been more extensive work on estimation and identification under
parametric or semi-parametric assumptions: see \citep{donald1996identification}
and \citep{athey2006empirical} for an overview.  For example,
\citet{athey2011comparing} fit the parameters of Weibull distributions to
observed maximum bids in to estimate bid distributions in USFS timber auctions.

\subsection{Preliminaries} 
\label{sec:prelims}
\paragraph{The (asymmetric) independent private values model.}
In this work, we consider the asymmetric independent private values (IPV)
model, with the additional stipulation that not all bids are observed. 
In this model, we observe a series of identical auctions between $k$ agents
(known $k \geq 2$): in each auction, every agent $i$ submits a {\em
bid} $X_i$, sampled independently from a (fixed) distribution with cumulative
distribution function $F_i$.
The sampled bids, together with the auction type, determine the {\em winner}
$Z$ of each auction (typically $Z = \arg\max_{i} X_i$) and a 
{\em transaction price} $Y$, i.e., what the winner pays for the auctioned
item.
In this work, we will only observe $Z$ and (sometimes) $Y$, and
rather than all bids $X_i$.

Two key differences between our setting and the typical IPV setup are (a)
the aforementioned partial observability; and (b) asymmetry---in particular, a
typical assumption is that all agents bid according 
to the {\em same} fixed distribution $F$, which simplifies both bid and value
estimation significantly (e.g., we could estimate a first-price auction by
learning the CDF of the largest bid, and then estimate the individual bid
distributions as the $k$-th root).

\paragraph{Statistical distances.}
Throughout our work, we provide finite-sample convergence bounds in terms of the
Wasserstein, L\'evy, and Kolmogorov distances, depending on the setting. 
The Wasserstein distance $\dw$
between two distributions $P$, $Q$ supported on $[0, 1]$ is 
\[\dw(P, Q) \triangleq \inf_{R} \E_{(x, y) \ts R} \left[ \abs{x - y} \right],\]
where the infimum is over all joint distributions $R$ with support $[0, 1]^2$ such
that the marginal of $x$ is equal to $P$ and the marginal of $y$ is equal to $Q$
when $(x, y) \ts R$. The Kolmogorov distance $\dk$ between two
distributions $P$ and $Q$ over an interval $I$ is defined as 
\[
    \dk(P, Q) \triangleq \sup_{x \in I}\ \abs{F_P(x) - F_Q(x)},
\] 
where $F_P$ and $F_Q$ are the cumulative distribution functions of $P$ and $Q$,
respectively. Finally, the L\'evy distance $D_L$ between $P$ and $Q$ is given by  
\[
    D_L(P, Q) = \min \lbrb{\epsilon:\ F_P(x - \epsilon) - \epsilon \leq F_Q(x) \leq F_P(x + \epsilon) + \epsilon}.
\]
Note that L\'evy distance is a strictly {\em weaker} notion than Wasserstein
distance and Kolmogorov distance, in the sense that both $D_L(P, Q) \leq \dk(P,
Q)$ and $D_L(P, Q) \leq \sqrt{\dw(P, Q)}$. Thus, all
of our results in Wasserstein and Kolmogorov distance
also effectively bound L\'evy distance.

\paragraph{Finite-sample rates.}
We present our convergence results using order notation:
in particular, the constants omitted from all order notation
in this paper are absolute constants that do not depend on the distributions
being estimated or any other parameters of the setting being considered. In
other words, our bounds are uniform. For example, whenever a bound in a theorem
statement reads as $O(f(k,1/\varepsilon,L))$, where $f$ is some function and $k,
\varepsilon, L$ are parameters of the setting, this means that there is an
absolute constant $C$ such that for any setting conforming to the setting
examined in this theorem we can replace $O(f(k,1/\varepsilon,L))$ in the theorem
statement by $C \cdot f(k,1/\varepsilon,L).$ Recall that
$\tilde{O}(f(\cdot))$ means that for some absolute constants $C$ and $k \ge 1$,
the bound can be replaced by $C f(\cdot) \log^k f(\cdot)$. Finally, in this work
we present most of our bounds in terms of the number of samples necessary to
attain a specific learning error $\varepsilon$---these can be straightforwardly
converted to bounds on $\varepsilon$ in terms of the number of samples $n$,
providing estimation rates as $n \to \infty$.

\section{Estimation from First-price Auction Data} 
\label{sec:max_sel_alg} 
In this section we show how to estimate the bid distributions from
a finite number of first-price auction observations.
We consider two regimes:
\begin{enumerate}
  \item[$\blacktriangleright$] In the \textit{full-support} regime, our goal is 
  to provide an estimation of the bid distributions in their whole support $[0, 1]$. As
  shown in \cref{thm:1stPrice:fullSupport}, in this regime we estimate the
  probability distributions within $\eps$ in Wasserstein distance. The sample 
  complexity here is $\approxeq (1/\eps)^k$ and has exponential
  dependence on the number of the agents $k$. As we explain in 
  \cref{sec:1stPrice:lowerBound}, this dependence on $k$ is necessary
  for the full-support regime (due to the exponentially low probability of
  observing winning bids near zero).
  \item[$\blacktriangleright$] In the \textit{effective-support} regime, our goal
  is to provide an estimation of the bid distributions only at the bid values that 
  have probability at least $\lambda$ to be observed as an outcome of the 
  first-price auction. As we show in \cref{thm:1stPrice:likely}, in this
  regime we avoid the exponential dependence on $k$ and we are able to get an
  algorithm that depends only polynomially in $\eps$ and $\gamma$ and only
  logarithmically in $k$. This is a doubly exponential improvement over the
  full-support regime and an exponential improvement on the best known
  effective-support result from \cite{blum2015learning}. 
  This result also provides the first algorithm with
  sublinear sample complexity for this problem.
\end{enumerate}

  Our first step is to formally define the procedure from which the first-price 
auction data are generated. The observation access that we assume is minimal in
the sense that we only observe the outcome of the auction; who wins and how much
they pay. 

\begin{definition}[First-Price Auction Data] \label{def:max_sel_obs}
    Let $\{F_i\}_{i = 1}^k$ be $k$ cumulative distribution functions with
  support $[0, 1]$, i.e. $F_i(x) = 0\ \forall x < 0$ and $F_i(1) = 1$. 
  A sample $(Y, Z)$ from a
  first-price auction with bid distributions $\{F_i\}_{i = 1}^k$ is generated as
  follows:
  \begin{enumerate}
    \item first generate $X_i \ts F_i$ independently for all $i \in [k]$,
    \item observe the tuple 
          $(Y, Z) \triangleq (\max_{i \in [k]} X_i, \argmax_{i \in [k]} X_i)$.
  \end{enumerate}
\end{definition}

A different access model explored in \cite{blum2015learning}
gives the econometrician control over an additional agent that has the ability
to bid arbitrarily, but only allows them to observe the identity of the winner
of each auction (not the transaction price). Using our result we can
also improve the result of \cite{blum2015learning} under this model (see
\cref{sec:1stPrice:additionalBidder}).

\subsection{Estimation of Bid Distributions} 
\label{sec:1stPrice:bids}
We are now ready to state our main results for the estimation of the bid 
distributions given sample access to first-price auction data as defined in 
\cref{def:max_sel_obs}. We start with our result for the full-support regime.

\begin{theorem}[First-Price Auctions -- Full Support] \label{thm:1stPrice:fullSupport}
  Let $\{(Y_i, Z_i)\}_{i = 1}^n$ be $n$ i.i.d. samples from the 
  first-price auction as per \cref{def:max_sel_obs}. Assume that the cumulative 
  distribution functions $F_i$ are continuous and satisfy 
  \( \abs{F_i(x) - F_i(y)} \ge \lambda \abs{x - y} \text{ for all } x, y \in [0, 1]. \)
  Then, there is a polynomial-time algorithm that computes
  functions $\hat{F}_i$ for $i \in [k]$ such that 
  \[ \Pr\left(\dw(\hat{F}_i, F_i) \le \eps\right) \ge 1 - \delta \] 
  for all $i \in [k]$ assuming that
  $n = \tilde{\Theta}\left(\lprp{\frac{2}{\lambda \cdot \eps}}^{4k}
  \frac{\log(1/\delta)}{\eps^2} \right)$, where $\dw$ is the Wasserstein distance.
\end{theorem}

  As we show in \cref{sec:1stPrice:lowerBound}, the sample complexity of 
\cref{thm:1stPrice:fullSupport} is almost optimal. Nevertheless, as we already 
explained, the exponential dependence on the number of agents $k$ can be reduced to
only logarithmic dependence if we only focus on the part of the support that is 
likely to be observed. In this case, our estimation guarantee is also 
simpler: we estimate the cumulative distribution functions 
with additive error $\eps$.

\begin{theorem}[First-Price Auctions -- Effective Support] \label{thm:1stPrice:likely}
  Let $\{(Y_i, Z_i)\}_{i = 1}^n$ be $n$ i.i.d. samples from the same 
  first-price auction as per \cref{def:max_sel_obs} and assume that the cumulative 
  distribution functions $F_i$ are continuous. Then, there exists a polynomial-time
  estimation algorithm, that computes the cumulative distribution functions $\hat{F}_i$
  for $i \in [k]$, such that for every
  $p, \gamma \in \{p, \gamma \ge 0:\ \Pr_{(Y, Z) \ts \calP_1}(Y \le p) \ge \gamma\}$, and 
  every $\eps \in (0, \gamma/2]$,
  \[ \Pr\left(\max_{x \in [p, 1]} \abs*{\hat{F}_i(x) - F_i(x)} \le \eps\right) \ge 1 - \delta \] 
  for all $i \in [k]$ assuming that 
  $n = \tilde{\Theta}\left({\log(k/\delta)}/{(\gamma^4 \eps^2)}\right)$.
\end{theorem}

Finally, we establish estimation of the corresponding probability density
functions $\{f_i\}$:
\begin{theorem} \label{thm:1stPrice:density}
  Let $\{(Y_i, Z_i)\}_{i = 1}^n$ be $n$ i.i.d. samples from the same 
  first-price auction as per \cref{def:max_sel_obs} and assume the
  densities $f_i$ of $F_i$ are well-defined and Lipschitz continuous, i.e.,
  \[ \abs*{f_i(x) - f_i(y)} \le L \abs*{x - y} \text{ for all } x, y \in [0, 1]. \]
  Then, there exists a polynomial-time estimation algorithm, that computes
  functions $\hat{f}_i$ for $i \in [k]$, with the following
  guarantee; for every $p, \gamma \ge 0$ such that 
  $\Pr(Y \le p) \ge \gamma$, and for every 
  $\eps \in (0, \gamma/2]$ it holds that 
  \[ \Pr\left( \int_p^1 \abs*{\hat{f}_i(x) - f_i(x)} dx \le \eps\right) \ge 1 - \delta \] 
  for all $i \in [k]$ assuming that 
  $n = \tilde{\Theta}\left({L^2 \cdot \log(k/\delta)}/{(\gamma^4 \cdot \eps^4)}\right)$.
\end{theorem}

Before explaining the formal proofs of the above theorems we give some intuition
behind our estimation algorithm. This intuition is given in a simplified setting 
where: (1) we assume the population model where we have access to infinitely many 
samples from the first-price auction data defined in \cref{def:max_sel_obs}, and
(2) the distributions are smoothed enough so that all the probability density 
functions that are involved are well defined. In this simplified setting we have
access to the following distributions:
\begin{enumerate}
    \item[$\triangleright$] $H_i$ is the cumulative distribution function of 
    $Y$ conditioned on 
    $Z = i$, for $i \in [k]$ and
    \item[$\triangleright$] $H$ is the cumulative distribution function of 
    $Y$ with no conditioning on $Z$.
\end{enumerate}
If we assume that all the $H_i$'s 
have well-defined densities $h_i$, then
\[ h_i(x) = f_i(x) \cdot \prod_{j \neq i} F_{j}(x), \quad H(x) = \prod_{j \in [k]} F_j(x), \quad \text{and} \quad H(x) = \sum_{j \in [k]} H_j(x). \]
Based on the above relations we can solve for the distribution $F_1$ as follows
\[ \frac{d}{dx} \log(F_i(x)) = \frac{f_i(x)}{F_i(x)} = \frac{h_i(x)}{H(x)} \implies F_i(x) = \exp\left( - \int_x^1 \frac{h_i(z)}{H(z)} d z \right). \]

This simple idea summarizes our approach in the population setting where 
infinite samples are available. Moving to the finite sample case an important
observation is that the aforementioned expression of $F_i$ can be also written
as
\[ F_i(x) = \exp\left(- \E_{(y, z) \ts \calP_1}\left[ \left. \frac{\bm{1}\{z = i\}}{H(y)} ~\right|~ y \ge x \right] \right). \]
The above expression allows the expectation involved to be estimated with an 
empirical expectation instead of an integral, assuming that a good estimation of $H(z)$ is 
computed. Towards designing our actual estimation algorithm and
proving its exact sample complexity we face the following additional technical
difficulties:
\begin{enumerate}
  \item in the above outline we assume that all the distributions are smooth 
  enough so that all the densities are well defined---in our main 
  theorem this assumption is not necessary,
  \item the usual estimation of $H$ has an additive error, whereas in the above 
  expression a multiplicative error guarantee is needed,
  \item the term $1/H(z)$ that is crucial in our estimation is not numerically
  stable for $z$ close to $0$ where $H(z)$ can also be very close to $0$ as 
  well.
\end{enumerate}

\subsubsection{Proof of \cref{thm:1stPrice:likely}} 
\label{sec:proof:1stPrice:likely}
We start by considering the effective-support setting.
Our first result will be an
information theoretic result enabling identification of $F_i$ with access to the
function $H$ and the measure $H_i$ (without requiring a density function).

\begin{lemma} \label{lem:it_ident}
    For all $i \in [k]$ and all $x \in (0, 1)$ such that $F_i(x) > 0$ and 
  $H(x) > 0$,
  \begin{equation*}
    F_i (x) = \exp \left(- \int_{x}^1 \frac{1}{H(y)} \ d H_i\right).
  \end{equation*}
\end{lemma}
\begin{proof}
  Using Lemma 3.1 from \cite{matematyczny2002chain} we have that
  \[
    \log(F_i(1)) - \log(F_i(x)) = \int_x^1 \frac{1}{F_i(y)} \ d F_i 
            = \int_x^1 \frac{\prod_{j \neq i} F_j(y)}{\prod_{j \in [k]} F_j(y)} \ d F_i 
            = \int_x^1 \frac{1}{H(y)} d H_i,
  \]
  where the last equality follows from the continuity of $F_i$'s and the properties
  of Riemann-Stieltjes integration. The lemma follows by observing that 
  $F_i(1) = 1$.
\end{proof}

We now focus our attention on obtaining good estimates of the quantity within the
exponential on the right hand side in \cref{lem:it_ident}. We introduce the
following notation: 
\begin{gather*}
    \hat{H}(x) \triangleq \frac{1}{n} \sum_{j = 1}^n \bm{1} \lbrb{Y_j \leq x} 
    \qquad 
    \hat{G}_i (x) \triangleq \frac{1}{n} \sum_{j = 1}^n \frac{1}{\hat{H} (Y_j)} \bm{1} \lbrb{Y_j \geq x \text{ and } Z_j = i}
\end{gather*}

\noindent Based on the above definitions we can define our estimate for $F_i$ as
\( \hat{F}_i (x) = \exp \left(-\hat{G}_i(x)\right). \)

\noindent Our next goal is to prove that $\hat{F}_i$ is close to $F_i$ for every
value $y \in [0, 1]$ such that $H(y) \ge \gamma$.

Now we establish 
concentration of $\hat{H}$. 
By the DKW inequality \cite{dkw}:
\begin{equation*}
  \max_{x \in [0, 1]} \abs*{\hat{H}(x) - H (x)} \leq \frac{1}{20} \cdot \gamma^2 \eps
\end{equation*}
with probability at least $1 - \delta / 2$ for our setting of $n$. Conditioning on
the above event and observing that $H(x) \geq \gamma$ for all $x \geq p$, 
\begin{equation} \label{eq:ih_err}
  \max_{x \in [p, 1]} \abs*{{1/H(x)} - {1/\hat{H}(x)}} \leq \frac{\eps}{10}.
\end{equation}
To establish concentration of $\hat{G}_i(x)$, we introduce another quantity
$\wt{G}_i$, defined as follows:
\begin{equation*}
  \wt{G}_i(x) \triangleq \frac{1}{n} \cdot \sum_{j = 1}^n \frac{1}{H(Y_i)} \bm{1} \lbrb{Y_j \geq x \text{ and } Z_j = i}.
\end{equation*}
We have from \eqref{eq:ih_err} that $\abs{\hat{G}_i (x) - \wt{G}_i (x)} \leq
\frac{\eps}{10}$ for all $x \in [p, 1]$.
Thus, it suffices to establish concentration of $\wt{G}$ around $G$. We first
prove concentration on a discrete set of points and interpolate to the rest of 
the interval. Define $U_i$ and $V_i$ as:
\begin{equation*}
  U_i = \left\{\gamma + i \cdot \frac{\eps}{10}: i \in [N] \cup \{0\} \text{ and } \gamma + i \cdot \frac{\eps}{10} \leq 1\right\} \cup \{1\} \text{ and } V_i = F_i^{-1} (U_i).
\end{equation*}
and let $V = \cup_{i \in [k]} V_i$. We have for all $x \in V, i \in [k]$, by
Hoeffding's inequality that:
    \begin{equation}
        \label{eq:gt_bound}
        \abs{\wt{G}_i (x) - G_i (x)} \leq \frac{\eps}{10}
        \quad 
        \text{with probability at least $1 - \delta / 2$.}
    \end{equation}
     We now condition on the above event as well. By combining \cref{eq:ih_err,eq:gt_bound}, we get:
    \begin{align*}
        \forall x \in V, i \in [k]: \abs{\hat{G}_i(x) - G_i(x)} 
        &\leq \eps/5, \text{and so for all $x \in V, i \in [k]$, we have:} \\
        \exp \lbrb{-{\eps/5}}\cdot F_i(x) \leq \hat{F}_i(x) 
        &\leq \exp \lbrb{{\eps/5}}\cdot F_i(x).
    \end{align*}
    We now extend from $V$ to the rest of $[p,1]$. Note that $\hat{G}_i(x)$ is a decreasing function of $x$. Hence, $\hat{F}_i$ is an increasing function of $x$. Now, let $x \in (p, 1) \setminus V$ and $i \in [k]$. We must have $x_l, x_h \in V$ with $x_l < x \leq x_h$ satisfying $F_i(x_h) - F_i(x_l) \leq \eps/10$. We now get:
    \begin{align*}
        &\hat{F}_i(x) \leq \hat{F}_i(x_h) \leq \exp\lbrb{{\eps/5}} 
        \cdot F_i(x_h) \leq \exp\lbrb{{\eps/5}} F_i(x_l) + 
        {\eps}/{8} \leq \exp\lbrb{{\eps/5}} F_i(x) + \eps/8,
        \\
        &\hat{F}_i(x) \geq 
        \exp\lbrb{-{\eps/5}} 
        F_i(x_l) \geq \exp\lbrb{-{\eps/5}} F_i(x_h) + 
        {\eps/10} \geq \exp\lbrb{- {\eps/5}} F_i(x) + {\eps/10}.
    \end{align*}
  The above two inequalities and our condition on $\eps$ conclude the proof. 
  \hfill \qed

\subsubsection{Proof of \cref{thm:1stPrice:fullSupport}} 
\label{sec:proof:1stPrice:fullSupport}
\noindent We now leverage our effective-support recovery result to recover bid
distributions on their full support (in Wasserstein distance). 
Under the ``lower bound on density'' assumption,
\begin{align} \label{eq:lowerBoundH}
  H(\eta) = \prod_{j \in [k]} F_j(\eta) \ge (\lambda \cdot \eta)^k.
\end{align}
Now, setting $\gamma = (\lambda \cdot \eta)^k$ and using
\cref{thm:1stPrice:likely} we have that 
$\tilde{\Theta}\left(\frac{\log(k/\delta)}{\lambda^k \cdot \eta^k \cdot \eta^2}\right)$ 
samples suffice to find estimates $\hat{F}_i$ such that the additive error
between $\hat{F}_i$ and $F_i$ is at most $\eta$ in the interval $[\eta, 1]$. For
every $i$, the maximum possible mass in the interval $[0, \eta]$ with respect to
the measure $F_i$ is $1$. Therefore, any two measures with support $[0, \eta]$ 
mass at most $1$ have a Wasserstein distance of at most $\eta$. Also, in the subset
$[\eta, 1]$ of the support we have that since the longest distance in the support 
is at most $1$ and $\max_{x \in [\eta, 1]} \abs*{\hat{F}_i(x) - F_i(x)} \le \eta$ 
we have that the Wasserstein distance of the measures $\hat{F}_i$ and $F_i$
conditioned on the support $[\eta, 1]$ is at most $\eps \cdot 1$. Thus,
\[ \dw(\hat{F}_i, F_i) \le 2 \cdot \eta. \]
Setting $\eta = \eps/2$ the theorem follows. \hfill \qed

\subsubsection{Proof of \cref{thm:1stPrice:density}} \label{sec:proof:1stPrice:density}

  We are going to use the estimation $\hat{F}_i$ from \cref{thm:1stPrice:likely}
together with the Lipschitzness of $f_i$ to prove this theorem. Let $h > 0$ and
$\epsilon_0 > 0$ be parameters that we will determine later. We define, for
every $x \in [p, 1]$, an density estimate
\[ \hat{f}_i(x) \triangleq \frac{1}{h} (\hat{F}_i(x + h) - \hat{F}_i(x)), \]
where due to \cref{thm:1stPrice:likely} we have 
$\abs{\hat{F}_i(x + h) - F_i(x + h)} \le \eps_0$ and 
$\abs{\hat{F}_i(x) - F_i(x)} \le \eps_0$ for $n =
\tilde{\Theta}\left(\frac{\log(k/\delta)}{\gamma^4 \eps^2}\right)$ samples.
Then,
\begin{align*}
  \int_p^1 \abs*{\hat{f}_i(x) - f_i(x)} \ d x 
    &= \int_p^1 \abs*{\frac{1}{h} (\hat{F}_i(x + h) - \hat{F}_i(x)) - f_i(x)} \ d x \\
    & \le \int_p^1 \abs*{\frac{1}{h} (F_i(x + h) - F_i(x)) - f_i(x)} \ d x + 2 \eps \\
    & = \int_p^1 \abs*{\frac{1}{h} \lprp{\int_{x}^{x + h} f_i(z) \ d z} - f_i(x)} \ d x + \frac{2 \eps}{h} \\
    &\le \int_p^1 \frac{1}{h} \lprp{\int_{x}^{x + h} \abs*{f_i(z) - f_i(x)} \ d z} \ d x + \frac{2 \eps}{h} \\
  \intertext{now due to the Lipschitzness of $f_i$ we have that}
  \int_p^1 \abs*{\hat{f}_i(x) - f_i(x)} \ d x 
    &\le \int_p^1 \frac{1}{h} \lprp{\int_{x}^{x + h} L \cdot \abs*{z - x} \ d z} \ d x + \frac{2 \eps}{h} \\
    & = \int_p^1 \frac{L}{h} \lprp{\frac{(x + h)^2}{2} - \frac{x^2}{2} - h \cdot x} \ d x + \frac{2 \eps}{h} \\ 
    &= \int_p^1 \frac{L}{h} \cdot h^2 \ d x + 2 \eps \le L \cdot h + \frac{2 \eps}{h}.
\end{align*}
Therefore, if we choose $h = \sqrt{\eps_0/L}$ and we also set $\eps = \eps_0^2/(9 L)$ 
the theorem follows.

\subsection{Lower Bound for Full-Support Estimation} \label{sec:1stPrice:lowerBound}

\newcommand{\unif}{\mrm{Unif}}

\noindent Here, we establish lower bounds proving the optimality of \cref{thm:1stPrice:fullSupport}. We prove:
\begin{enumerate}
    \item the exponential dependence on $k$ incurred in \cref{thm:1stPrice:fullSupport} is necessary and
    \item the distributions cannot be recovered in Kolmogorov distance in their whole support.
\end{enumerate}
In both these cases, we will construct a pair of distributions $\{f_i\}_{i = 1}^k$ and $\{f^\prime_i\}_{i = 1}^k$ satisfying the bounded density condition of \cref{thm:1stPrice:fullSupport} such that:
\begin{enumerate}
    \item $f_1$ and $f^\prime_1$ have $\dw (f_1, f^\prime_1) \geq \Omega(\eps)$ and $\dk (f_1, f^\prime_1) \geq 1 / 2$ and 
    \item Fewer than $\Omega((\lambda \eps)^{-(k - 1)})$ fail to distinguish them with large probability.
\end{enumerate}
The main intuition behind our construction 
is that learning the
behavior of any of the densities below $\eps$ requires observing  
$Y \leq \eps$ and this only happens with probability $\eps^{-k}$.  

\begin{theorem}
    \label{thm:main_lb}
    Let $k \in \mb{N}$, and let $\eps, \lambda \in (0, 1/2)$. 
    Then, there exist two tuples of distributions $\mc{D} = \{f_i\}_{i = 1}^k$ and 
    $\mc{D}^\prime = \{f^\prime_i\}_{i = 1}^k$ 
    with the first price auction model (\cref{def:max_sel_obs}) 
    on $\mc{D}, \mc{D}^\prime$ satisfies the density bound condition from
    \cref{thm:1stPrice:fullSupport} such that for any estimator $\hat{\mu}$, we
    have: 
    \begin{equation*}
        \max \lprp{
        \mb{P} \lbrb{\dw \lprp{\hat{\mu} (\lbrb{(Y_i, Z_i)}_{i = 1}^n), \mc{D}} \geq \frac{\eps}{8}}, 
        \mb{P} \lbrb{\dw \lprp{\hat{\mu} (\lbrb{(Y_i^\prime, Z_i^\prime)}_{i = 1}^n), \mc{D}^\prime} \geq \frac{\eps}{8}}} \geq \frac{1}{3}
    \end{equation*}
    where $(Y_i, Z_i)$, $(Y^\prime_i, Z^\prime_i)$ are drawn i.i.d from
    $\mc{D}$ and $\mc{D}^\prime$ respectively if $n \leq
    \frac{1}{10} \cdot (\lambda \eps)^{-(k - 1)}$. 
\end{theorem}
\begin{proof}
    Let $\mc{D}_1 = \{f_1, \dots, f_k\}$ and $\mc{D}_2 = \{f^\prime_1, \dots,
    f^\prime_k\}$ denote the two sets of distributions characterizing our
    first price auction model (\cref{def:max_sel_obs}). 
    We will have $f_i = f^\prime_i$ for all $i > 1$.
    \begin{equation*}
        f^\prime_i = f_i = \lambda \cdot \unif ([0, 1]) + (1 - \lambda) \cdot \unif([3/4, 1])
        \quad 
        \text{ for all } i > 1.
    \end{equation*}
    However, $f_1$ and $f_1^\prime$ will have large Wasserstein and Kolmogorov
    distance:
    \begin{align*}
        f_1 &= \lambda \cdot \unif ([0, 1]) + 
              (1 - \lambda) \cdot \unif ([0,\eps / 4]) \\
        f^\prime_1 &= \lambda \cdot \unif ([0, 1]) + (1 - \lambda) \cdot \unif ([3\eps / 4, \eps]).
    \end{align*}
    Let $(Y, Z)$ and $(Y^\prime, Z^\prime)$ be distributed according to the
    first price auction model with respect to $\mc{D}_1$ and $\mc{D}_2$. 
    We now define the events $E$ and $E'$ on $(Y, Z)$ and $(Y', Z')$ as follows:
    \begin{equation}
      \label{eq:pe_bnd}
        E = \{Y \in (\eps, 1]\} \text{ and } E^\prime = \{Y^\prime \in (\eps, 1]\} \implies \P \lbrb{E} = \P \lbrb{E^\prime} = 1 - (\lambda \eps)^{k - 1}
    \end{equation}
    By construction, $(Y,Z)$ and $(Y^\prime,Z^\prime)$ have the same
    distribution conditioned on $E$ and $E'$.
    
    Now, let $\bm{W} = \{(Y_i, Z_i)\}_{i \in [n]}$ and $\bm{W}^\prime = \{(Y_i,
    Z_i)\}_{i = 1}^n$ be collections of $n$ i.i.d samples from
    $\mc{D}$ and $\mc{D}^\prime$ respectively, and let $\hat{\mu}$ denote any
    estimator of the first price auction model. We show that $\hat{\mu}$ has
    large error on at least one of $\mc{D}$ or $\mc{D}^\prime$. Letting $F$
    (respectively, $F^\prime$) denote the event that $E$ (respectively,
    $E^\prime$) holds for all of the $(Y_i, Z_i)$ (respectively, $(Y^\prime_i,
    Z^\prime_i)$), we have:  
    \begin{align*}
        \mb{P} \lr{\dw(\hat{\mu} (\bm{W}), \mc{D}) \leq \nicefrac{\eps}{8}} 
        = \mb{P}(F) \cdot \mb{P}\lr{\dw(\hat{\mu} (\bm{W}), \mc{D}) \leq
        \nicefrac{\eps}{8} \vert F} 
        + \mb{P} (\bar{F}) \cdot \mb{P} \lr{\dw(\hat{\mu} (\bm{W}), \mc{D}) \leq \nicefrac{\eps}{8} | \bar{F}}.
    \end{align*}
    Now, if $n \leq \frac{1}{10} (\lambda
    \eps)^{-(k-1)}$, we have from \cref{eq:pe_bnd} and a union bound that
    $\mb{P}(F) \geq 9/10$. 
    Furthermore, note that conditioned on $F$ and $F^\prime$,
    $\bm{W}$ and $\bm{W}^\prime$ have the same distribution and $\dw(f_1,
    f_1^\prime) \geq \eps / 4$. Assuming the probability in the above equation
    is greater than $2 / 3$, we may re-arrange the above equation as follows:
    \begin{align*}
        \frac{2}{3} \leq \mb{P} (F) \mb{P} \lbrb{\dw(\hat{\mu} (\bm{W}), \mc{D}) \leq \frac{\eps}{8} \biggr| F} + \frac{1}{10} \leq \mb{P} (F^\prime) \mb{P} \lbrb{\dw(\hat{\mu} (\bm{W}^\prime), \mc{D}^\prime) \geq \frac{\eps}{8} \biggr| F^\prime} + \frac{1}{10}.
    \end{align*}
    By re-arranging the above equation, we have that either:
    \begin{align*}
        \mb{P} \lbrb{\dw(\hat{\mu} (\bm{W}), \mc{D}) \leq \frac{\eps}{8}} \leq \frac{2}{3},
        \text{ or }
        \mb{P} \lbrb{\dw(\hat{\mu} (\bm{W}^\prime), \mc{D}^\prime) \leq \frac{\eps}{8}} \leq \frac{2}{3}
    \end{align*}
    concluding the proof of the theorem.
\end{proof}
Note that the probabilities $1 / 3$ chosen in the above theorem is not a
substantial restriction as any algorithm successfully distinguishing between
$\mc{D}$ and $\mc{D}^\prime$ with probability bounded away from
$\nicefrac{1}{2}$ can be boosted to arbitrarily high probability by simple
repetition. As a simple consequence of this construction, we can 
rule out estimation in Kolmogorov distance:

\begin{theorem}
    \label{thm:lb_kol}
    Let $n \in \mb{N}$ and $\hat{\mu}$ be an estimator for the First-Price-Auction model. Then, for all $\delta > 0$, there exists a First-Price-Auction model characterized by $\mc{D} = \{f_i\}_{i = 1}^k$ satisfying the bounded density condition of \cref{thm:1stPrice:fullSupport} satisfying:
    \begin{equation*}
        \mb{P} \lprp{\dk (\hat{\mu} (\bm{W}), \mc{D}) \leq \frac{1}{4}} \leq \frac{1}{2} + \delta.
    \end{equation*}
    where $\bm{W} = \{(Y_i, Z_i)\}_{i = 1}^n$ are drawn i.i.d from the first
    price auction model on $\mc{D}$.
\end{theorem}
\begin{proof}
  We will prove the lemma via contradiction. Let $n, \hat{\mu}$ be such that the there exists $\delta > 0$ such that for all First-Price-Auction models, $\mc{D}$, satisfying the bounded density condition:
  \begin{equation*}
      \mb{P} \lprp{\dk (\hat{\mu} (\bm{W}), \mc{D}) \leq \frac{1}{4}} \geq \frac{1}{2} + \delta.
  \end{equation*}
  Note that by repeating the experiment $\Omega(1 / \delta^2)$ times, we may boost the success probability to $9 / 10$ by taking the pointwise median of the resulting estimates. However, from our construction in the proof of \cref{thm:main_lb}, we have by picking $\eps$ small enough in the construction that there exists a distribution, $\mc{D}$ such that:
  \begin{equation*}
      \mb{P} \lbrb{\dk \lprp{\hat{\mu} (\bm{W}), \mc{D}} \geq \frac{1}{4}} \geq \frac{1}{3}
  \end{equation*}
  as all the distributions we construct have Kolmogorov distance greater than $1/2$ between them. This yields the contradiction, proving the theorem.
\end{proof}

\subsection{Estimation from Partial Observations} 
\label{sec:1stPrice:additionalBidder}

  In this section we show how our results in the previous sections can be 
translated to the partial observation model introduced by \cite{blum2015learning} 
defined below.

\begin{definition}[Partial Observation Data] \label{def:max_partial_obs}
    Let $\{F_i\}_{i = 1}^k$ be $k$ cumulative distribution functions with support 
  $[0, 1]$, i.e. $F_i(x) = 0\ \forall x < 0$ and $F_i(1) = 1$. A sample $(r, Y, Z)$ from a
  first-price auction with bid distributions $\{F_i\}_{i = 1}^k$ is generated as
  follows:
  \begin{enumerate}
    \item we, the observer, pick a price $r \in [0, 1]$, and let $X_{k + 1} = r$
    \item generate $X_i \ts F_i$ independently for all $i \in [k]$,
    \item observe a winner $Z = \argmax_{i \in [k+1]} X_i$.
  \end{enumerate}
\end{definition}

At first glance, it seems like the access to partial observation data is more
restrictive than the access to the first-price auction data that we defined in 
\cref{def:max_sel_obs}. Nevertheless, we show that partial observations 
suffice to run the same 
estimation used in \cref{sec:1stPrice:bids}. 

\begin{theorem}[First-Price Auctions -- Partial Observations]
  \label{thm:fp_partial_obs}
  Let $\{Z_i\}_{i = 1}^n$ be $n$ i.i.d. partially observed samples from the same 
  first-price auction as per \cref{def:max_partial_obs} and assume that the cumulative 
  distribution functions $F_i$ are continuous and admit Lipschitz-continuous
  densities $f_i$ with constant $L$. Then, given $p, \gamma \in [0, 1]$ such
  that $\Pr(X_{i \in [k]} \leq p) \geq \gamma$, 
  there exists a polynomial-time
  estimation algorithm, that computes the cumulative distribution functions $\hat{F}_i$
  for $i \in [k]$, so that for every $\eps \in (0, \gamma/2]$ it holds that 
  \[ \Pr\left(\max_{x \in [p, 1]} \abs*{\hat{F}_i(x) - F_i(x)} \le \eps\right) \ge 1 - \delta \] 
  for all $i \in [k]$ assuming that 
  $n = \Theta\lr{
    \frac{k}{\gamma^6\epsilon^5}\log\lr{\frac{k}{\gamma^2\epsilon\alpha}}\log\lr{\frac{L}{\gamma^2 \epsilon}}
  }$
\end{theorem}
\begin{proof}
Note that under the partial observation model, we can estimate
$\Pr(Y \geq x)$ for any fixed $x$ by setting the reserve price to $x$ (i.e., bidding
$x$) and counting the number of times that the planted bid wins the auction.
More precisely, we can define %
\[
   \hat{H}(x) = 1 - \frac{1}{n_1} \sum_{i=1}^{n_1} \mathbf{1} \lbrb{Z_i = k+1}.
\]

We can similarly define, for any given agent $i$, an estimator for the
probability that the agent wins with price less than or equal to $x$:
\[
    \hat{H}_i (x) \triangleq 
    \frac{1}{n_2} \sum_{j=1}^{n_2} \bm{1}
        \lbrb{Z_j = i \text{ at reserve price } 0}
    - 
    \frac{1}{n_2} \sum_{j = 1}^{n_2} \bm{1} 
        \lbrb{Z_j = i \text{ at reserve price $x$}}.
\]

\newcommand{\ygap}{\delta}

By construction $\hat{H}_i \rightarrow H_i$ and $\hat{H} \rightarrow H$ where
$H_i$ and $H$ are as defined earlier. Similarly to our strategy in the proof of
\cref{thm:1stPrice:likely}, let
\begin{equation*}
    U = \left\{\gamma + i \cdot \ygap : i \in \mathbb{N} \cup \{0\} \text{ and } \gamma + i \cdot \ygap \leq 1\right\} \cup \{1\} 
    \text{ and } V = H^{-1} (U)
    \text{ and } W = H_i^{-1} (U)
\end{equation*}

\newcommand{\numsteps}{T}
\newcommand{\numquantiles}{N}
\newcommand{\Hquants}[1]{v_{{#1}}}
\newcommand{\estHquants}[1]{\hat{v}_{{#1}}}
\newcommand{\Hiquants}[2]{w_{{#2}}}
\newcommand{\estHiquants}[2]{\hat{w}_{{#2}}}
\newcommand{\yquants}[1]{u_{{#1}}}
\newcommand{\hoeffgap}{\beta}
\newcommand{\binsearchgap}{\epsilon_1}
\newcommand{\binHstepgap}{\epsilon_1/2}
\newcommand{\binHistepgap}{\epsilon_1/2}

For convenience, define $\numquantiles = |U| \leq \ygap^{-1}$, and recall that
$k$ is the number of agents in the auction.
Our first goal is to obtain a set of estimates $\estHquants{j} \approx
\Hquants{j}$ for the quantiles of $H$ and another set $\estHiquants{i}{j}
\approx \Hiquants{i}{j}$ for the quantiles of each $H_i$.
To accomplish this, we will run, for each $\yquants{j} \in U$, 
$\numsteps$ iterations of binary search between $0$ and $1$. In particular, 
we initialize $\estHquants{j}^{(0)} = 1$ then, for each successive iteration
$t$, a Hoeffding bound shows that 
\begin{align}
\label{eq:fp_pointwise_conc}
\Pr\lr{ \abs*{\hat{H}(\estHquants{j}^{(t)}) - H(\estHquants{j}^{(t)})} \geq \binHstepgap} 
< 2\exp\lbrb{-2 \lr{\binHstepgap}^2 n}. \\
\Pr\lr{ \abs*{\hat{H}_i(\estHiquants{i}{j}^{(t)}) - H_i(\estHiquants{i}{j}^{(t)})} \geq \binHistepgap} 
< 2\exp\lbrb{-\lr{\binHistepgap}^2 n/2}.
\end{align}
We condition on the above events by taking a union bound over all agents, all
search steps, and all points $u_i \in U$.
Now, for each iteration $t$ of the binary search:
\begin{enumerate}
    \item If $|\hat{H}(\hat{v}_i^{(t)}) - u_t| \leq \binHstepgap$,
    then $|H(\hat{v}_i^{(t)}) - u_t| \leq \binsearchgap$, so we terminate
    and set $\hat{v}_i = \hat{v}_i^{(t)}$,
    \item otherwise if $\hat{H}(\hat{v}_i^{(t)}) - u_t > \binHstepgap$ then
    $H(\hat{v}_i^{(t)}) > u_t$ and we search the upper interval,
    \item otherwise $\hat{H}(\hat{v}_i^{(t)}) - u_t < -\binHstepgap$ and so
    $H(\hat{v}_i^{(t)}) < u_t$ and we search the lower interval.
\end{enumerate}
We perform an analogous process to find the $\estHiquants{i}{j}$.
This ensures the correctness of the binary search, and
setting $\numsteps = \log(2L/\binsearchgap)$, where $L$ is a Lipschitz
constant of $H$ and all $H_i$, guarantees that after performing
this search for each $u_i$, we will find $\hat{V}$ and $\hat{W}$ such that
\begin{align}
    \label{eq:cond1_partial_info}
    |H(\hat{v}_j) - u_j| &\leq \binsearchgap \qquad 
        \text{for all}\qquad j \in [|U|], \text{ and} \\
    |H_i(\estHiquants{i}{j}) - u_j| &\leq \binsearchgap \qquad 
        \text{for all}\qquad j \in [|U|] \text{ and } i \in [k] \\
    \nonumber
    \text{w.p.} \qquad & 1 - 2 \numsteps N \exp \lbrb{-2(\binHstepgap)^2} - 2k \numsteps N \exp \lbrb{-(\binHistepgap)^2/2} 
\end{align}

\noindent In order to define our approximation of $G_i$, we will
consider the list of indices $X = V \cup W_i$, i.e., the union of the
estimated quantiles of $H$ and $H_i$. Using Hoeffding's inequality,
\begin{align*}
    |H(x_j) - \hat{H}(x_j)| &\leq \hoeffgap \qquad \text{for all}\qquad j \in [|X|], \text{ and} \\
    |H_i(x_j) - \hat{H}_i(x_j)| &\leq \hoeffgap \qquad \text{for all}\qquad j \in [|X|], \text{ and} \\
    \nonumber
    \text{w.p.} \qquad & 1 - 4kN \exp\lbrb{2n \hoeffgap^2} \geq 1 - \frac{4k}{\delta} \exp\lbrb{2n \hoeffgap^2}
\end{align*}

We further condition on the above and define an estimate of $G_i(x_j) 
= \int_{x_j}^1 \frac{1}{H(z)}\, dH_i(z)$,
\[
    \hat{G}_i(x_j) = \sum_{s=j}^{|X|-1} \lr{\hat{H}_i(x_{s+1}) - \hat{H}_i(x_s)}/{\hat{H}(x_s)}.
\]
Using the mean value theorem (see Appendix \ref{sec:fp_partial_info_extras} for more detail), 
\begin{align*}
    \abs{G_i(\hat{v}_t) - \hat{G}_i(\hat{v}_t)} 
    &\leq 
    \frac{2}{\gamma} \cdot \sum_{s=1}^{|X|} \abs{H_i(x_{s}) - \hat{H}_i(x_{s})}
    + \max_{s \in [|X|]}\ \ \abs*{
        \frac{1}{H(x_{s+1})} - \frac{1}{\hat{H}(x_s)}
    } \\
    &\leq 
    \frac{2|X|\hoeffgap}{\gamma} + \frac{\ygap + 2 \cdot \binsearchgap + \hoeffgap}{\gamma^2}
\end{align*}

We now extend our approximation from the set of points $\{x_i\}$ to the entire
interval $[\rho, 1]$. Note that for any $x$ in this interval, there exists an
$x_h, x_{h+1} \in X$ such that $x \in (x_h, x_{h+1}]$. Furthermore, both $G_i$
and $\hat{G}_i$ are monotonic in $x$ by construction, and 
\[
    \abs*{G_i(x_{h+1}) - G_i(x_h)} = \abs*{\int_{x_h}^{x_{h+1}} \frac{1}{H(z)}\, dH_i} \leq \frac{1}{\gamma} \cdot (\ygap + 2\binsearchgap),
\]
since the $x_i$ are at least as close as the quantiles of $H_i$, while
$\frac{1}{H(z)} \leq \frac{1}{\gamma}$ by assumption. Using this inequality and
the monotonicity of $G_i$ yields: 
\begin{align*}
    G_i(x) \geq G_i(x_h) 
        &\geq \hat{G}_i(x_{h+1}) - \frac{1}{\gamma} (\ygap + 2\binsearchgap) - \frac{2|X|\hoeffgap \gamma + \ygap + 2 \binsearchgap + \hoeffgap}{\gamma^2} \\
           &\geq \hat{G}_i(x_{h+1}) - \frac{(2|X| \gamma + 1)\hoeffgap + (\gamma + 1)(\ygap + 2\binsearchgap)}{\gamma^2} \\
           &\geq \hat{G}_i(x) - \frac{4(k+1) \gamma \hoeffgap/\ygap + 2\ygap + 4\binsearchgap}{\gamma^2} \\
    \text{Analogously, } G_i(x) &\leq \hat{G}_i(x) + \frac{4(k+1) \gamma \hoeffgap / \ygap + 2\ygap + 4\binsearchgap}{\gamma^2}
\end{align*}
Now, set:
\(
    \ygap = \frac{\gamma^2\epsilon}{6},
    \binsearchgap = \frac{\gamma^2 \epsilon}{24},
\) and \(
    \hoeffgap = \frac{\gamma\epsilon^2}{24(k+1)},
\)
so that
\(
    |\hat{G}_i(x) - G_i(x)| \leq \frac{\epsilon}{2}
\)
with probability 
\begin{align*}
    &1 - \frac{12}{\gamma^2\epsilon}\log\lr{\frac{48L}{\gamma^2 \epsilon}}\exp \lbrb{- \frac{\gamma^4\epsilon^2}{1152} n}
      - \frac{12k}{\gamma^2\epsilon}\log\lr{\frac{48L}{\gamma^2 \epsilon}}\exp \lbrb{- \frac{\gamma^4\epsilon^2}{4608} n}
      - \frac{24k}{\gamma^2\epsilon} \exp\lbrb{-\frac{\gamma^2 \epsilon^4}{288(k+1)^2} n}.
\end{align*}
Thus, setting
$$n = \frac{4608(k+1)^2}{\gamma^4\epsilon^4}\log\lr{\frac{3}{\alpha}
\frac{24k}{\gamma^2\epsilon} \log\lr{\frac{48L}{\gamma^2 \epsilon}}}$$
makes this probability $1 - \alpha$. Now, the total number of samples required
for this approach is $O(N \cdot k \cdot T \cdot n)$, concluding the proof.
\end{proof}
\begin{remark}[Inserting bids may change equilibria]
  As we highlighted above, using access to the partial observation data we can
  estimate the distributions $H_i$ to within $\epsilon$ error and thus apply a
  similar algorithm to the ordinary first-price setting. 
  Observe, however, that by inserting arbitrary bids to get good
  estimates of the functions $F_i$, the econometrician can affect the bidding
  strategy of the agents and thus interfere with the equilibrium point 
  of the first-price auction. 
  This is not true for our model in
  \cref{def:max_sel_obs}, where the econometrician is a passive observer (in
  particular, observations do not interfere with the 
  equilibrium of the agents) and hence the bid distributions can lead
  to an estimation of the value distributions as well (as we show in
  \cref{sec:value_estimation}). 
\end{remark}

\subsection{Estimation of Value Distributions}
\label{sec:value_estimation}
Theorems \ref{thm:1stPrice:fullSupport}, \ref{thm:1stPrice:likely}, and
\ref{thm:1stPrice:density} establish recovery results for the bid distribution
of each agent. In a first-price auction at (Bayes-Nash) equilibrium, however,
these bid distributions do not correspond to agents' value distributions.
Instead, at equilibrium, each agent draws a value $v_i \sim G_i(\cdot)$ and bids
the best responses to other agents, i.e.,  
\begin{equation}
    \label{eq:bid_argmax} 
    \beta_i(v_i) = \arg\max_b u_i(b; v_i) \coloneqq \arg\max_b\ (v_i - b) \prod_{j \neq i} F_i(b).
\end{equation}
As discussed in the introduction, in our asymmetric IPV setting,
we (a) cannot write an explicit form for the optimal
bid for a given value; and (b) cannot derive smoothness results for the bid
distribution from smoothness assumptions on the value distribution. In fact, a 
unique Bayes-Nash equilibrium is not even guaranteed to exist.

\paragraph{Approach.} \citet{lebrun2006uniqueness} provides the following
characterization of Bayes-Nash equilibria for asymmetric first-price auctions,
and shows that such equilibria exist and are unique under
under some (relatively mild) assumptions:

\renewcommand{\bi}[1]{\beta_i(#1)}
\newcommand{\estbi}[1]{\widehat{\beta_i}(#1)}
\renewcommand{\bj}[1]{\beta_j(#1)}
\newcommand{\vi}[1]{\alpha_i(#1)}
\newcommand{\vj}[1]{\alpha_j(#1)}
\newcommand{\vip}[1]{\alpha_i'(#1)}
\newcommand{\vjp}[1]{\alpha_j'(#1)}
\newcommand{\Gj}[1]{G_j(\vj{#1})}
\newcommand{\Gi}[1]{G_i(\vi{#1})}
\newcommand{\gj}[1]{g_j(\vi{#1})}
\newcommand{\gi}[1]{g_i(\vi{#1})}
\newcommand{\qi}{\frac{d}{db} \log\, G_i(\vi{b})}
\newcommand{\qj}{\frac{d}{db} \log\, G_j(\vj{b})}

\begin{lemma}[\citet{lebrun2006uniqueness}]
    \label{lemma:characterization_equilibrium}
    Suppose the agents' values are distributed according to the right-continuous
    cumulative distribution functions $G_i(\cdot)$ with support $[0, 1]$ and
    whose derivatives (i.e., the value density functions) $g_i(\cdot)$ are
    locally bounded away from zero. Then, a set of strategies (bid functions)
    $\vi{\cdot}: [0, 1] \rightarrow [0, \infty)$ is a Bayesian equilibrium if and only
    if there exists an $\eta \in [0, 1]$ such that the inverses $\vi{\cdot} =
    \beta_i^{-1}(\cdot)$ exist, are strictly increasing, and form a solution over
    $[0, \eta]$ of the following system of differential equations:
    \begin{align}
        \label{eq:equilibrium_characterization}
        \qi &= \frac{1}{n-1} \lr{\frac{-(n-2)}{\vi{b} - b} + \sum_{j \neq i} \frac{1}{\vj{b} - b}},\ 
        \vi{0} = 0,\ \vi{\eta} = 1.
    \end{align}
    If bidders are not permitted to bid above their values, and if one of the
    following two conditions are met, then the set of strategies $\bi{\cdot}$
    represents a unique equilibrium: 
    \begin{enumerate}
        \item[(i)] The value distributions have an atom at zero, i.e., $G_i(0) > 0$, 
        \item[(ii)] There exists $\delta > 0$ such that the cumulative density
        function of the $i$-th agent's value is strictly log-concave over $(0,
        \delta)$ for all $i$. 
    \end{enumerate}
\end{lemma}
\begin{corollary}
    \label{cor:convenient_equilibrium}
    Rearranging Equation
    \eqref{eq:equilibrium_characterization} from
    \cref{lemma:characterization_equilibrium} yields
    \begin{align}
        \label{eq:log_deriv_diff}
        \sum_{j \neq i} \qj = \frac{1}{\vi{b} - b},
        \quad
        \text{ and in turn }
        \quad 
        \sum_{j \neq i} \frac{f_j(b)}{F_j(b)} = \frac{1}{\vi{b} - b}.
    \end{align}
\end{corollary}

Since Lemma \ref{lemma:characterization_equilibrium} guarantees that the inverse
bidding strategies are strictly increasing at equilibrium, the inverse mapping
theorem dictates that 
    $
    G_i(v) = F_i(\bi{v}).
    $
If the equilibrium strategies $\bi{\cdot}$ were known, we could apply
our results for the bid distributions to estimate $G_i(v)$ directly.
In our setting, however, we do not have access to
the strategies $\bi{\cdot}$---in fact, a general closed form does not
exist for asymmetric auctions.

Instead, we will use the characterization given by Equation
\eqref{eq:bid_argmax} of the equilibrium bid as the best response to other
bidders. In particular, since we have accurate estimates for each $F_i(\cdot)$,
we can define the following empirical versions of each quantity
introduced so far, including $\widehat{G_i}$, an estimate for the cumulative
distribution functions of each agent's value:
\begin{align}
    \label{eq:approx_utility}
    \widehat{u_i}(b; v_i) = (v_i - b)\, \prod_{j \neq i} \widehat{F_j}(b),
    \quad
    \estbi{v} = \arg\max_{b} \widehat{u_i}(b; v_i),
    \quad
    \widehat{G_i}(v) = \widehat{F_i}(\estbi{v}).
\end{align}
Turning this into a formal argument requires tackling the following technical
challenges:
\begin{enumerate}
    \item Approximating the utility function via $\widehat{u_i}(\cdot; v_i)$ and
    efficiently maximizing the estimated utility function to find the optimal bid.
    \item Showing that the maximizer of the
    $\widehat{u_i}(\cdot, v_i)$ is close to that of the true utility.
    \item Bounding the combined error incurred from our empirical approximations.
\end{enumerate}

We will start by tackling the above challenges in the 
effective-support regime. The following characterizes the setup as
well as the additional assumptions used to estimate the value distribution:
\begin{assumption}[Value Estimation]
    \label{assumption:value_est_1}
    We assume the preconditions of \cref{lemma:characterization_equilibrium}, as
    well as an upper bound on the density of the value distributions, i.e., 
    $g_i(b) \leq \zeta \text{ for all } i \in [k].$
\end{assumption}
\begin{definition}[Effective Support]
    \label{label:value_assumption}
    In the effective-support setting,
    we are given $(p, \gamma) \in [0, 1]$ such that 
    \(
        \prod_{i \in [k]} F_i(p) \geq \gamma.
    \)
    This is identical to the effective-support setting for bid estimation, with
    the addition that $p$ is pre-defined, since it is part of the
    estimation algorithm.  
\end{definition}

We are now prepared to tackle recovery of the valuation distribution in the
effective-support regime. Our main result is captured in
\cref{thm:value_estimation} below. After proving the result, we will show how it
straightforwardly extends to full-support estimation, in a similar manner to our
bid estimation results.
\begin{theorem}[Estimation of Value Distributions -- Effective Support]
    \label{thm:value_estimation}
    Let $\{(Y_i, Z_i)\}_{i=1}^n$ be $n$ i.i.d. samples from the same first-price
    auction as per \cref{def:max_sel_obs}. 
    Under the setup of \cref{label:value_assumption}, there exists a
    polynomial-time estimation algorithm that computes cumulative distribution 
    functions $\hat{G}_i(\cdot)$ for $i \in [k]$ with the following guarantee:
    \begin{align*}
        \sup_{v \in [p, 1]} \abs*{\hat{G}_i(v) - G_i(v)} &\leq \epsilon
        &\text{ if } n = 
        \tilde{\Theta}\lr{\frac{k^2 \zeta^2 L^6 \log(1/\delta)}{\gamma^{10} \epsilon^6}}
        \text{ and } F_i \text{ is $L$-Lipschitz} \\
        D_L\lr{
            \hat{G}_i \cdot \bm{1}_{[p, 1]}, G_i \cdot \bm{1}_{[p, 1]}
        } &\leq \epsilon
        &\text{ if } n = 
        \tilde{\Theta}\lr{\frac{k^2 \zeta^2 \log(1/\delta)}{\gamma^{16} \epsilon^{12}}}
        \phantom{an \text{ and } F_i \text{is $L$-Lipschitz}} 
    \end{align*}
    for all $i \in [k]$, where $D_L(\cdot, \cdot)$ is distance in the L\'evy metric.
\end{theorem}
\begin{remark}[Testing Lipschitzness]
    The first guarantee given by \cref{thm:value_estimation} depends on the
    (global) Lipschitzness of the bid CDFs $F_i$.
    While efficiently testing global Lipschitzness of $F_i$ from samples is
    impossible, if we have a specific $\epsilon_0 > 0$ in mind, we can instead
    test 
    \[
        \widehat{L} = \max_{|x - y| = \epsilon_0}\, \frac{1}{\epsilon_0} \lr{\abs*{\widehat{F}_i(x) - \widehat{F}_i(y)} + 2\epsilon}.
    \]
    This maximization can be done efficiently in $n$ steps, since
    $\widehat{F}_i$ is piecewise constant. 
    Since we condition on accurate bid CDF estimation, we have that for any
    $x$ and $y$,  
    \begin{align*}
        \abs*{\widehat{F}_i(x) - \widehat{F}_i(y)} \geq \abs*{F_i(x) - F_i(y)} - 2\epsilon,
    \end{align*}
    and so $\widehat{L} \geq \max_{|x - y| = \epsilon_0}
    \frac{1}{\epsilon_0}\abs{F_i(x) - F_i(y)}$, and so we can use
    $\widehat{L}$ in place of $L$ in the theorem.
\end{remark}

\begin{proof}[Proof of Theorem \ref{thm:value_estimation}]
In this effective-support setup,
\cref{thm:1stPrice:likely} guarantees that with probability $1 - \delta$, we can
learn the bid CDFs on the interval $[\bi{\rho}, 1]$ up to additive error
$\epsilon_0$, for any $\epsilon_0 > 0$, in  
\(
    n = \tilde{\Theta}\lr{\frac{\log(k/\delta)}{\gamma^4 \epsilon_0^2}}
\)
queries. We thus assume that we have CDF estimates $\widehat{F_i}(\cdot)$ that
are within $\epsilon_0$ of the corresponding true CDFs in the effective
support. 

Our point of start is to show that we can estimate the approximate utility
function efficiently.
For each agent $i$, we can relabel each observed
data point $(Y, Z)$ as $(Y, \bm{1}_{Z = i})$ and run our estimation
procedure on the corresponding two-agent auction to get piecewise-constant
$\epsilon_0$-approximations of $\prod_{j \neq i} F_j(b)$ for all $i \in [k]$.
Since $v_i, b \in [0, 1]$, we can condition on the event that
$\widehat{u_i}(\cdot, v_i)$ is an $\epsilon_0$-approximate estimate of the true
utility function.

Now, the form of our estimate for $\prod_{j \neq i} F_j(b)$ is
piecewise constant (with $n$ pieces, where $n$ is the sample complexity of
\cref{thm:1stPrice:likely}) 
and monotonically increasing in $b$.
Meanwhile, $(v_i - b)$ is strictly decreasing in $b$ along any
interval. 
We can thus
exactly maximize $\widehat{u_i}$ by evaluating it at $n$ locations (i.e., the
beginning of each piecewise-constant interval).

Next, define $b^*$ and $\widehat{b}$ to be the maximizer of $u_i(\cdot;
v_i)$ and $\widehat{u_i}(\cdot; v_i)$ respectively over the interval $[p, 1]$.
Since we have an $\epsilon_0$-approximation of utility within this interval,
\begin{align}
    \label{eq:uspace-closeness}
    \abs*{
        u_i(b^*; v_i) - u_i(\widehat{b}; v_i)
    } &\leq 2\epsilon_0.
\end{align}

Our next goal is to translate this proximity in utility-space to proximity in
parameter-space, i.e., to show that $b^* \approx \widehat{b}$. To do so, we use
the derivative of the utility function with respect to the bid, which is given
by the following result:
    \begin{lemma}[Derivative of utility function]
        \label{lem:utility_derivative}
        Fix any $v_i \in [0, 1]$, and let $b^* = \bi{v_i}$ be the equilibrium bid
        for the $i$-th agent corresponding to value $v_i$. Let $u_i(\cdot; v_i)$
        denote the utility function as in \eqref{eq:bid_argmax}. Then, 
        \begin{align*}
            \frac{d}{db} u_i(b; v) = \lr{\vi{b^*} - \vi{b}} \sum_{j \neq i} \frac{f_j(b)}{F_j(b)} \prod_{j \neq i} F_j(b).
        \end{align*}
    \end{lemma}
    \begin{proof}
        Recalling the definition of the utility function from \eqref{eq:bid_argmax},
        \begin{align*}
            u_i(b; v) &= (v - b) \cdot \prod_{j\neq i} F_j(b) \\
            \frac{d}{db} u_i(b; v) &= -\prod_{j \neq i} F_j(b) + (v - b) \sum_{j \neq i} f_j(b) \prod_{k \neq j,i} F_k(b) \\
            &= \left(
                    (v - \vi{b}) \sum_{j \neq i} \frac{f_j(b)}{F_j(b)} +
                    (\vi{b} - b) \sum_{j \neq i} \frac{f_j(b)}{F_j(b)} - 1
                \right) \prod_{j \neq i} F_j(b).
        \end{align*}
        Observing that $v_i = \vi{b^*}$ and using
        \cref{cor:convenient_equilibrium} concludes the proof. 
    \end{proof}
    Now, returning to \eqref{eq:uspace-closeness},
    \begin{align}
        \label{eq:eps_bound_ve}
        2\epsilon_0 &\geq 
        \abs*{\int_{\widehat{b}}^{b^*} \lr{\vi{b^*} - \vi{x}} 
            \cdot \frac{1}{\vi{x} - x} \cdot \prod_{j \neq i} F_j(x)\, dx}.
        &\text{(\cref{cor:convenient_equilibrium})}
    \end{align}
    In order to bound $\abs{b^* - \widehat{b}},$ we need
    the following lower bound on the derivative $\vip{x}$:
    \begin{restatable}{lemma}{qilb}
        \label{lem:qi_lb}
        Under the conditions of \cref{assumption:value_est_1},
        for all $b \in [\rho, 1]$,
        \begin{align*}
            \qi > L(b) \qquad \text{ where }
            L(b) \coloneqq \frac{\vi{b} - b}{(k-1)^2 \cdot \zeta}.
        \end{align*}
    \end{restatable}
    We defer the proof of \cref{lem:qi_lb} to the Online Appendix (the proof is
    nearly identical to that of Lemma A-1 in \citep{lebrun2006uniqueness}, with
    the exception that we keep better track of constants to get a non-zero
    lower bound), and state the corollary:
    \begin{corollary}
        Under the conditions of \cref{lem:qi_lb},
        \begin{align*}
            \vip{b} = \frac{\gamma}{\zeta} \cdot \lr{ \qi } \geq \frac{\gamma}{(k-1)^2 \cdot \zeta^2} \cdot (\vi{b} - b).
        \end{align*}
    \end{corollary}
Now, let $\bar{b} = (\widehat{b} + b^*)/2.$ By positivity of the integrand in \eqref{eq:eps_bound_ve},
    \begin{align*}
        2\epsilon_0 
        &\geq \abs*{\int_{\bar{b}}^{\widehat{b}} 
            \lr{\vi{x} - \vi{\bar{b}}}
            \cdot \frac{1}{\vi{x} - x} \cdot \prod_{j \neq i} F_j(x)\, dx}.
    \end{align*}
    Using the intermediate value theorem, there exists $z \in [\min(x, \bar{b}), \max(x, \bar{b})]$ so that 
    \begin{align*}
        2\epsilon_0 &\geq \abs*{\int_{\bar{b}}^{\widehat{b}} 
            \frac{\lr{x - \bar{b}} \vip{z}}{\vi{x} - x} \cdot \prod_{j \neq i} F_j(x)\, dx}
        \geq \int_{\bar{b}}^{\widehat{b}} 
            \frac{\gamma\abs{x-\bar{b}}}{(k-1)^2 \cdot \zeta^2}
            \cdot \frac{\vi{z} - z}{\vi{x} - x} \cdot \prod_{j \neq i} F_j(x)\, dx.
    \end{align*}
    Using our effective support definition and bid optimality,
    \begin{align*}
        \lr{\vi{z} - z} \cdot \prod_{j \neq i} F_j(z) 
            &\geq \lr{\vi{z} - p} \cdot \gamma 
            \geq \lr{z - p} \cdot \gamma 
            \geq \lr{\Delta/2} \cdot \gamma,
    \end{align*} 
    where $\Delta = \abs{\widehat{b} - b^*}$. Thus,
    since $\vi{x} - x \leq 1$,
    \begin{align*}
        \frac{\vi{z} - z}{\vi{x} - x} \cdot \prod_{j \neq i} F_j(x) 
        &\geq \lr{\Delta/2} \cdot \gamma \cdot \lr{\prod_{j \neq i} F_j(x) / \prod_{j \neq i} F_j(z)} 
        \geq \lr{\Delta/2} \cdot \gamma^2 
    \end{align*}
    Returning to the integral,
    \begin{align}
        2\epsilon_0
        &\geq \frac{\Delta \gamma^3}{2 (k-1)^2 \cdot \zeta^2} \cdot \int_{\bar{b}}^{\widehat{b}} 
            \abs{x - \bar{b}} \, dx 
        \label{eq:case2}
        \geq \frac{\Delta \gamma^3}{2 (k-1)^2 \cdot \zeta^2} \cdot \lr{\frac{\Delta^2}{8}}.
    \end{align}
    Thus, for any $\epsilon_1 > 0$, setting 
    \(
        \epsilon_0 = \frac{\epsilon_1^3 \Gamma^3}{32k^2\zeta^2}
    \)
    implies that $\Delta < \epsilon_1$.

    We are now ready to bound the error in our estimate of the valuation
    distribution. We consider two cases: first, when the bid CDFs $F_i(\cdot)$
    satisfy a Lipschitz-like constraint; second, a more general setting where we
    only require a lower bound on the valuation densities $g_i(\cdot)$. In the
    first case, we learn the valuation distributions in Kolmogorov distance over the
    interval $[\rho, 1]$, whereas in the second case we learn in 1-Wasserstein
    distance. 

    In both cases, our estimate will be given by:
    \begin{align*}
        \widehat{G}_i(v) \coloneqq \bm{1}\lbrb{v \in [p, 1]}\cdot \widehat{F}_i(\widehat{b}(v))  
    \end{align*}

    \paragraph{{\bf Case 1}: Lipschitz bid CDF} In the first case, we
    assume that the cumulative distribution function of each bid distribution is
    $L$-Lipschitz continuous
    (note that any bid distribution with density bounded by $L$ satisfies this).
    Then, for any $v \in [p, 1]$, 
    \begin{align*}
            \abs*{G_i(v) - \widehat{F}_i(\widehat{b}(v))} 
            &\leq \abs*{G_i(v) - F_i(b(v))} + \abs*{F_i(\widehat{b}(v)) - F_i(b(v))} + \epsilon_0 
            \leq L \cdot \epsilon_1 + \epsilon_0 < 2L\epsilon_1,
    \end{align*}
    where we recall that $\widehat{G}_i(v) \coloneqq
    \widehat{F}_i(\widehat{b}(v))$. 
    Thus, we define $\epsilon_1 = \epsilon / (2L)$ so that 
    \[
        \sup_{v \in [p, 1]} \abs*{G_i(v) - \widehat{G}_i(v)} \leq \epsilon
        \text{ using }
        n 
        = \tilde{\Theta}\lr{\frac{k^2 \zeta^2 \log(1/\delta)}{\gamma^{10} \epsilon_1^6}}
        = \tilde{\Theta}\lr{\frac{k^2 \zeta^2 L^6 \log(1/\delta)}{\gamma^{10} \epsilon^6}}
    \]
    samples, which concludes the proof. 

    \paragraph{{\bf Case 2}: General case.} We can also obtain a convergence
    guarantee that is independent of the Lipschitz continuity of $F_i$ by
    considering a slightly smaller interval $[p + d, 1]$ for some $d > 0$. In
    this setting, the ``best response'' property of bids implies that, for any
    $b \in [p + d, 1]$,
    \begin{align*}
        \lr{\vi{b} - b} \cdot \prod_{j \neq i} F_j(b) 
            &\geq \lr{\vi{b} - p} \cdot \gamma 
            \geq \lr{b - p} \cdot \gamma 
            \geq d \cdot \gamma.
    \end{align*}
    Thus from \cref{cor:convenient_equilibrium},
    \begin{align*}
        \sum_{j \neq i} \frac{f_j(b)}{F_j(b)} 
            = \frac{1}{\vi{b} - b} 
            \leq \frac{1}{d \cdot \gamma}
        \implies
        f_j(b) \leq \frac{1}{d\cdot \gamma} \text{ for all $j \in [k]$.}
    \end{align*}
    Using the same method as Case I (with $L = 1/(d\cdot
    \gamma)$), we can guarantee that for $d > 0$,
    \begin{align}
        \label{eq:value_nolipschitz_bd}
        \sup_{v \in [p + d, 1]}\ \abs*{\widehat{G}_i(v) - G_i(v)} < \epsilon,
        \text{ as long as }
        n \in \tilde{\Theta}\lr{\frac{k^2 \zeta^2 \log(1/\delta)}{\gamma^{16} d^6 \epsilon^6}}.
    \end{align}
    Setting $d = \epsilon$ concludes the proof of the Theorem.
\end{proof}
As a consequence of \cref{thm:value_estimation}, we can accurately estimate
valuation distributions in Wasserstein distance when the bid densities are lower
bounded:

\begin{theorem}[Estimation of Value Distributions -- Full Support]
    \label{thm:value_estimation_full}
    Let $\{(Y_i, Z_i)\}_{i=1}^n$ be $n$ i.i.d. samples from the same first-price
    auction as per \cref{def:max_sel_obs}. 
    Under the setup of \cref{label:value_assumption} and assuming
    \( \abs{F_i(x) - F_i(y)} \ge \lambda \abs{x - y} \text{ for all } x, y \in [0, 1], \)
    there exists a polynomial-time estimation algorithm that
    computes cumulative distribution functions $\hat{G}_i(\cdot)$ for $i \in
    [k]$ with the following guarantee for all $i \in [k]$:  
    \begin{align*}
        \dw\lr{\hat{G}_i(v), G_i(v)} &\leq \epsilon
        &\text{ if } n = 
        \tilde{\Theta}\lr{\lr{\frac{1024}{\lambda^{10} \epsilon^{10}}}^k {\zeta^2 L^6 \log(1/\delta)}}
        \text{ and } F_i \text{ is $L$-Lipschitz,} \\
        \dw\lr{\hat{G}_i(v), G_i(v)} &\leq \epsilon
        &\text{ if } n = 
        \tilde{\Theta}\lr{\lr{\frac{2048^2 \cdot (11/8)^{16}\cdot (11/3)^6}{\lambda^{16} \epsilon^{22}}}^k \zeta^2 \log(1/\delta)}
        \text{ otherwise.}
    \end{align*}
\end{theorem}
\begin{proof}
We proceed identically to the proof of
\cref{thm:1stPrice:fullSupport}. In particular, for any $\eta > 0$,
\[
    H(\eta) = \prod_{j \in [k]} F_j(\eta) \ge (\lambda \cdot \eta)^k.
\]
Set $\gamma = (\lambda \cdot \eta)^k$ and we 
use the first case of theorem \cref{thm:value_estimation}, such that with
$$\tilde{\Theta}\left(\frac{k^2 \zeta^2 L^6 \log(1/\delta)}{\lambda^{10k} \cdot \eta^{10k} \cdot \eta^6}\right)$$ 
samples we can find estimates $\hat{G}_i$ for all $i$ such that the additive error
between $\hat{G}_i$ and $G_i$ is at most $\eta$ in the interval $[\eta, 1]$.
From here an identical argument to that of the proof of
\cref{thm:1stPrice:fullSupport} (i.e.,
\cref{sec:proof:1stPrice:fullSupport}) shows that
\( \dw(\hat{G}_i, G_i) \le 2 \cdot \eta, \)
after which setting $\eta = \eps/2$ the first case in the theorem follows.
For the second case, we use \eqref{eq:value_nolipschitz_bd} with 
$p = 8\eta/11$, $d = 3\eta/11$, $\gamma = (8\lambda \eta/11)^k$, and $\eta =
\epsilon/2$. 
\end{proof}

\section{Estimation from Second-price Auction Data}
\label{sec:sp_sel_alg}
In this section, we will state and prove our main result for the estimation of
bid distributions from second-price auction observations. Unlike the
first-price-auction setting, our main result in this setting involves estimating
the bid distributions under the 
\emph{full-support} regime where we aim to obtain distributions approximating
$F_i$ up to small error in Kolmogorov distance. As in the first-price-setting,
this incurs an exponential dependence on $k$. We leave the problem of estimation
in the \emph{effective-support} regime as an open problem for future work.
(We do show in \cref{sec:sp_bid_insert} that if the econometrician can
insert bids of their own, then even just the identity of the winner and an
indicator of whether the reserve price was paid suffice to estimate bid
distributions over the effective support [\cref{thm:sp_bid_insert}], using a
very simple algorithm.) 
The identification of the cumulative density functions, $F_i$, given access to the
distribution of $(Y, W)$ was previously established in
\cite{athey2002identification} building on techniques from reliability theory
\cite{meilijson1981estimation}. However, this is to our knowledge, the first
result establishing non-parametric finite-sample recovery from observations of
second-price-auction data. We now formalize our considered observation model:
\begin{definition}
    \label{def:second_price_sel}
    Let $\{f_i\}_{i = 1}^k$ be $k$ probability density functions on $[0, 1]$. An
    observation from the second-price selection model on $\{f_i\}_{i = 1}^k$ is
    defined as follows:
    \begin{enumerate}
        \item First generate $X_i \ts f_i$ independently for $i \in [k]$
        \item Define $W \coloneqq \argmax_{i \in [k]} X_i$
        \item Observe the tuple $(Y, W)$ where $Y \coloneqq \max_{i \in [k]
        \setminus \{W\}} X_i$.
    \end{enumerate}
\end{definition}

We now state the assumptions on $F_i$ required for our guarantees to hold:
\begin{assumption}
    \label{as:second_price_bdd}
    The distributions $F_i$ admit densities $f_i$ with
    $\alpha \leq f_i \leq \eta$ for constants $\alpha, \eta > 0$.
\end{assumption}
Note that in comparison to \cref{thm:1stPrice:fullSupport}, we require an upper
bound on the densities, $f_i$, in addition to the lower bound property used
previously. Our main result in this setting is the following theorem where we
establish efficient, finite sample recovery guarantees from second-price-auction
data satisfying \cref{as:second_price_bdd}:
\begin{theorem}
    \label{thm:sec_price_fin_samp}
    Let $\eps \in (0, 1)$ and $\bm{X} = \{(Y_i, W_i)\}_{i = 1}^n$ denote $n$ i.i.d observations from a Second-Price-Selection model (\cref{def:second_price_sel}) satisfying \cref{as:second_price_bdd}. Then, it is possible to learn in polynomial time cumulative distribution functions $\hat{F}_i$ satisfying:
    \begin{equation*}
        \sup_{x \in [0, 1]} \abs{\hat{F}_i (x) - F_i (x)} \leq \eps 
        \text{ with probability at least } 1 - \rho
    \end{equation*}
    as long as $\eps \leq e^{- C^1_{\eta, \alpha} k}$ and $n \geq
    \lprp{\frac{1}{\eps}}^{C^2_{\eta, \alpha} k} \log 1 / \rho$ for some
   absolute constants $C^1_{\eta, \alpha}, C^2_{\eta, \alpha}$.
\end{theorem}

In the rest of the section, we prove \cref{thm:sec_price_fin_samp}. In \cref{ssec:sec_price_approach}, we present a high level overview of our proof strategy where analogously to the first-price case, we derive a differential equation relating the densities, $f_i$ to the distribution of $(Y, W)$ (\cref{def:second_price_sel}) leading to a fixed point equation satisfied by the cumulative functions $\{F_i\}$. Subsequently, in \cref{ssec:fp_def}, we present a discretized version of the fixed point iteration that we analyze to prove our theorem. We carry out the formal analysis of our in \cref{ssec:sec_price_proof}. 

\subsection{Approach}
\label{ssec:sec_price_approach}
We now provide a high-level overview of our proof of
\cref{thm:sec_price_fin_samp}. For clarity, we will first outline the proof in
the idealized population (infinite-sample) setting. In the next section, we
formalize this outline and establish finite-sample guarantees.

We start by deriving a fixed point equation which plays a central part in our
analysis. 
Note that in the population setting, we have access to the functions 
\begin{equation}
    \label{eq:sp_gi_def}
    G_i(x) = \Pr\lr{W = i, Y \leq x}.
\end{equation}
For each $G_i(x)$, we can use the independence of the bid distributions
$F_i$ and some simple calculations to define a valid corresponding density:
\[
    g_i(x) = \lr{1 - F_i(x)} \sum_{j \neq i} f_j(x) \prod_{l\neq i,j} F_l(x).
\]
Rearranging and taking advantage of the product rule allows us to simplify this
as
\[
    \prod_{j \neq i} F_j(x) = \int_0^x \frac{1}{1-F_i(z)} g_i(z)\, dz.
\]
Re-parameterizing the above by letting $U^*_i =
\prod_{j \neq i} F_i$, we obtain the fixed point equation:
\[
    U^*_i(x) = \int_0^x \frac{1}{1-H_i(U^* (z))} g_i(z)\, dz \  \text{ where }
        H_{i}(v) = \frac{\lprp{{\prod_{j\neq i} v_j^{1/(k-1)}}} }{v_i^{(k-2)/(k-1)}} \\
\]

We divide the domain into smaller intervals and approximate $\wt{F}_i$ by a piecewise-constant function on each. 
The Banach fixed-point theorem \citep{banach} is crucial to our analysis:
\begin{theorem}[Banach Fixed-Point Theorem \citep{banach,patafpt}]
    \label{thm:banach}
    Let $(X, d)$ be a complete metric space with mapping $T: X \rightarrow X$ contractive with respect to $d$; i.e there exists $\rho < 1$ such that: 
    \begin{equation*}
        \forall x, y \in X: d(T(x), T(y)) \leq \rho \cdot d(x, y).
    \end{equation*}
    Then, $T$ admits a unique fixed-point $x^* \in X$. Furthermore, for any $x_0 \in X$, the sequence $\{x_n\}_{n \in \mb{N}}$, defined by $x_n \coloneqq T(x_{n - 1})$ for all $n \geq 1$, satisfies:
    \begin{equation*}
        \forall n \geq 0: d(x_n, x^*) \leq \rho^n \cdot d(x_0, x^*).
    \end{equation*} 
\end{theorem}

Unfortunately, it turns out that the fixed point iteration just described is {\em not} contractive with respect to its input---instead, we proceed iteratively, starting with the origin and estimating $\wt{F}_i$ for each successive interval in turn. After solving for a set of intervals, we treat them as fixed, and construct a new fixed-point iteration for the next interval. We make this argument precise in the following sections. In doing so, 
    the key technical difficulty faced by our approach is in the \emph{amplification} of errors incurred at earlier stages of the algorithm into later stages. As we shall see, errors due to approximation, sampling or computation are compounded exponentially over the running of the algorithm (\cref{lem:err_prop}). Therefore, we use a careful data-based approach where samples are used to decide the widths of successive intervals, ensuring both that the fixed point equation remains contractive and crucially, that there are not too many stages where successive intervals are constructed (\cref{lem:T_bnd}).

\newcommand{\secsampcomp}{C \lprp{\frac{\eta}{\alpha}}^{4k^2} \frac{1}{\eps^{10k}} k (\log \eta / \alpha + \log 1 / \delta \eps)}
\newcommand{\tq}{\wt{Q}}
\newcommand{\tv}{{V^*}}

\newcommand{\hg}{\hat{G}}
\newcommand{\hgerr}{\frac{\eps}{100}}
\newcommand{\hq}{\widehat{Q}}
\newcommand{\finsampint}{\alpha \xi}
\newcommand{\finsampintx}{\xi}
\newcommand{\finphi}{\widehat{\phi}}
\newcommand{\finfix}{\hat{\bm{v}}}

\newcommand{\prevest}{\widehat{V}^{(t-1)}}
\newcommand{\prevtrue}{V^{(t-1)}}
\newcommand{\vst}{\bm{v}^*}

\subsection{Fixed Point Definition}
\label{ssec:fp_def}

Here, we formally define the version of the fixed point iteration used in our algorithm. 
Recall that the CDFs of the bid generating distributions, $F_i$, satisfy:
\begin{equation*}
    \forall i \in [k]: U^*_i (x) \coloneqq \prod_{j \neq i} F_j (x) = \int_{0}^x \frac{1}{1 - F_i (z)} \cdot g_i (z) dz. 
\end{equation*}
Recasting the above equation in terms of the functions, $U^*_i$, we obtain:
\begin{equation}
    \forall i \in [k]: U^*_i (x) \coloneqq \prod_{j \neq i} F_j (x) = \int_{0}^x \frac{g_i (z)}{1 - H_i (z)} dz \text{ with } H_i (z) = \frac{\prod_{j \neq i} (U^*_j (z))^{1 / (k - 1)}}{(U^*_i (z))^{(k - 2) / (k - 1)}}.
    \tag{FP-CONT}
    \label{eq:fixed_point}
\end{equation}
In our algorithm, we approximate the functions, $U^*_i$, by piecewise constant functions on intervals of width $\delta$ (a parameter which we set later, in \ref{eq:par_set}) and approximate solutions to the above fixed-point equation. We will prove that the approximation errors as well as errors due to computational and statistical  constraints remain small despite these choices.

In the finite sample setting, we do not have access to the precise population functions, $G_i$ and hence, approximate them with their empirical counterparts:
\begin{equation*}
    \forall x \in [0, 1]: \hat{G}_i (x) \coloneqq \frac{1}{n} \cdot \sum_{j = 1}^n \bm{1} \lbrb{W_j = i, Y_j \leq x}
\end{equation*}

Before running our main algorithm, we first compute (very) coarse approximations
$\smash{\hat{U}_i(\cdot)}$ to the functions $\smash{U_i^*(\cdot)}$; it turns out
(see \cref{alg:est_u} and \cref{lem:u_appx})
that we can compute such loose approximations efficiently as a simple
consequence of the DKW inequality \citep{dkw}.

Our algorithm then operates in stages: we divide the interval $[\nu, 1 - \theta]$
(we set $\nu$ and $\theta$ later, in \eqref{eq:par_set}) into 
a finite number of ``macro-intervals'', each of which
contains a number of micro-intervals of width $\delta$. The macro-intervals are
defined by endpoints $\nu \coloneqq x_{0} < x_{1} \dots < x_{T} \leq 1 - \theta / 2$, where $T$ and the width of each macro-interval are chosen dynamically based on observed data. Within each macro-interval $[x_{\tau - 1}, x_{\tau}]$ are $\ell^{(\tau)}$ micro-intervals of width $\delta$, defined by endpoints $x_{\tau - 1} \coloneqq x_{\tau,0} < x_{\tau,1} < \ldots < x_{\tau,\ell^{(\tau)}} \eqqcolon x_{\tau}$. Picking the endpoints $\{x_{\tau}\}_{\tau \in [T]}$ requires balancing contractivity of the fixed-point iteration with the overall number of macro-intervals which determine the accuracy of the final solution. We formally describe our algorithm for computing the endpoints in \cref{alg:macro_ints}. To ease notation in the remainder of the section, we use $\Delta^{(\tau)}_{i, m}$ to denote differentials of $\hat{G}_i$ across the grid points and extend the definition of $x_{\tau, \ell}$:
\begin{gather*}
    \forall i \in [k], \tau \in [T], l \in [\ell^{(\tau)}]: \Delta^{(\tau)}_{i, l} \coloneqq \hat{G}_i (x_{\tau, l}) - \hat{G}_i (x_{\tau, l - 1}) \\
    \forall \tau \in [T], \ell \in \mb{N} \cup \{0\}: x_{\tau, \ell} \coloneqq x_{\tau - 1} + \ell \cdot \delta.
\end{gather*}

The resulting intervals define a discretization of the fixed-point equation
\eqref{eq:fixed_point}. As we previously discussed, our algorithm operates in
stages: at stage $\tau$, we assume access to an estimate $\smash{V_i^{(\tau)}
\approx U^*_i(x_{\tau,0})}$ from the previous stage 
(recall that $\smash{x_{\tau,0} = x_{\tau-1,\ell^{(\tau-1)}}}$ by definition).
We instantiate $\smash{U^{(\tau)} \in \mathbb{R}^{k \times \ell^{(\tau)}}}$ 
where $\smash{U^{(\tau)}_{i, l}}$ is designed to approximate $U^*_i (x_{\tau,
l})$.
We estimate the elements of $\smash{U^{(\tau)}}$ via the following discretized
fixed-point iteration: 
\begin{align*}
    \phi^{(\tau)}_{i, l} (U^{(\tau)}) 
        &= \mathrm{clip} \lprp{
                \sum_{m = 1}^l \frac{1}{1 - H^{(\tau)}_{i, m}(U^{(\tau)})} \cdot \Delta^{(\tau)}_{i, m} + V^{(\tau)}_{i}
                ,\frac{1}{2\eta} \cdot \hat{U}_i (x_{\tau, l}), \frac{2}{\alpha} \cdot \hat{U}_i (x_{\tau, l})} \\
    H^{(\tau)}_{i, m} (U^{(\tau)}) &= \mathrm{clip} \lprp{
        \frac{\prod_{j \neq i} (U^{(\tau)}_{j, m})^{\frac{1}{(k - 1)}}}{(U^{(\tau)}_{i, m})^{\frac{(k - 2)}{(k - 1)}}}
        ,\alpha x_{\tau, m}, \min \lprp{1 - \alpha (1 - x_{\tau, m}), \eta x_{\tau, m}}}. \tag{FP} \label{eq:fp_fin_def}
\end{align*}
In \eqref{eq:fp_fin_def}, $\hat{U}_i(\cdot)$ are the aforementioned coarse
approximations to $U_i^*(\cdot)$, and $V_i^{(\tau)}$ is the estimate of $U^*_i
(x_{\tau - 1})$ obtained by running $L$ iterations of the fixed point
iteration $\phi^{(\tau - 1)}$ on the previous macro-interval (where $L$ is set in
\eqref{eq:par_set}). We initialize the algorithm with
\[
    V_i^{(0)} \coloneqq \hat{G}_i (\nu).
\]

We present a summary of the entire algorithm corresponding to
\cref{thm:sec_price_fin_samp} in \cref{alg:sp_estimation}, and spend the
remainder of this section proving that this algorithm yields an
$\eps$-approximation to $U^*_i(\cdot)$ for all $i \in [k]$: 
The precise choices of our parameters are provided below:
\begin{gather*}
    \theta \coloneqq \frac{\eps}{16 \eta},
    \quad \delta \coloneqq \lprp{\frac{\alpha}{8\eta} \nu}^{32k}, 
    \quad \eps_g \coloneqq \delta \cdot \lprp{\frac{\alpha \nu}{8 \eta}}^{32k}, 
    \quad L \coloneqq \log (4 / \eps_g) \\
    \nu \coloneqq \min 
    \lbrb{\lprp{\frac{\alpha}{2\eta}}^{256}, \exp \lprp{- 2^{32} k \lprp{\frac{\eta}{\alpha}} \log \lprp{\frac{2\eta}{\alpha}}}, \lprp{\frac{\theta}{2}}^{\lprp{\frac{4\eta}{\alpha}}^{16}}, \lprp{\frac{\alpha \eps}{32 \eta}}^{24}
    } \tag{PAR} \label{eq:par_set}
\end{gather*}
The parameter $\theta$ is chosen to be suitable small such that $F_i (x)$ is well approximated by $0$ when $x \leq \theta$ and $1$ when $x \geq 1 - \theta$. Hence, we only need to obtain good estimates in the range $[\theta, 1 - \theta]$. However, as shown in \cref{sssec:err_prop}, our estimation errors grow exponentially in the number of macro-intervals, $T$ but drop rapidly as our start point decreases (\cref{lem:err_prop}). Consequently, $\nu$ is picked to be significantly smaller than $\theta$ to ensure our starting error is small. The parameters $\delta$ and $\eps_g$ are then chosen small with respect to $\nu$ to ensure the approximation errors caused by the discretization of our fixed point iteration and statistical errors from approximating the functions $G_i(\cdot)$ with $\hat{G}_i$ remain small. Finally, the parameter which determines how many times we iterate our fixed point equation, $L$, is picked to control errors due to computational limits.

\begin{algorithm}[h]
    \caption{Our algorithm for estimating value distributions in second-price auctions.}
    \label{alg:sp_estimation}
    \begin{algorithmic}[1]
    \Procedure{Estimate-Second-Price}{$\{(Y_j, Z_j)\}_{j=1}^n, \delta, \alpha, \eta, \nu, \theta$}
       \State $\{x_{\tau, m}: \tau \in [T], m \in [\ell^{(\tau)}]\} \gets $
       \Call{Construct-Macro-Intervals}{$\{(Y_j, Z_j)\}_{j=1}^n, \delta, \alpha,
       \eta, \nu, \theta$}
       \For{$i \in \{1,\ldots,k\}$}
       \State $V_i^{(0)} \coloneqq \hat{G}_i (\nu)$
       \State $\hat{U}_i(\cdot) \gets $ \Call{Coarse-Approximate-$U$}{
        $\{(Y_j, Z_j)\}_{j=1}^n$, $i$} \Comment{See \cref{lem:u_appx}}
        \State $U_i(x) \gets \hat{G}_i (\nu) \cdot \bm{1} \lbrb{x \leq \nu}$ \Comment{Estimate of $U^*_i$}
       \EndFor
       \For{$\tau \in \{1,\ldots T\}$}
       \State Initialize $U^{(\tau)} \in \mathbb{R}^{k \times \ell^{(\tau)}}$
       such that $U^{(\tau)}_{i,m} \gets \hat{U}_i(x_{\tau,m})$
       \Comment{Initialize with coarse approx.}
       \State Define $\phi^{(\tau)}: \mathbb{R}^{k \times \ell^{(\tau)}} \to
       \mathbb{R}^{k \times \ell^{(\tau)}}$ as in \eqref{eq:fp_fin_def}
       \State \textbf{repeat $L$ times}: $U^{(\tau)} \gets \phi^{(\tau)}(U^{(\tau)})$
       \State $U_i(x) \gets U_i(x) + \sum_{m=1}^{\ell^{(\tau)}} U^{(\tau)}_{i,m} \cdot
       \bm{1}\lbrb{x \in [x_{\tau,m-1}, x_{\tau,m})}$
       \Comment{Extend our estimate}
       \State $V^{(\tau+1)}_{i} \gets U^{(\tau)}_{i,\ell^{(\tau)}}$
       \EndFor
       \State \Return $\{U_i\}_{i=1}^k$
    \EndProcedure
    \end{algorithmic}
\end{algorithm}

\subsection{Proof of \cref{thm:sec_price_fin_samp}}
\label{ssec:sec_price_proof}
In this subsection, we prove \cref{thm:sec_price_fin_samp}. The sole
probabilistic condition we require is the empirical concentration of the
functions $\hat{G}_i$ around $G_i$ (defined in \eqref{eq:sp_gi_def}), 
which we establish below:
\begin{lemma}
    \label{lem:epsg_lem}
    We have, for $\eps_g$ as defined in \eqref{eq:par_set},
    as long as $n \geq \log (2 k / \rho) / (2\eps_g^2)$,
    \begin{equation*}
        \mb{P} \lbrb{\forall i \in [k]: \norm{\hat{G}_i - G_i}_\infty \leq \eps_g} \geq 1 - \rho.
    \end{equation*}
\end{lemma}
\begin{proof}
    For each $i \in [k]$ and $j \in [n]$, define the random variables, 
    $\smash{A^i_j \coloneqq Y_j \cdot \bm{1} \{Z_j = i\} + \bm{1} \{Z_j \neq i\}}$. 
    The CDFs of the $A^i_j$ variables correspond to $G_i$, while their empirical
    CDFs correspond to $\hat{G}_i$. The lemma thus follows from the DKW
    inequality \citep{dkw} and a union bound over $i$.   
\end{proof}

\noindent For the the rest of the proof, we condition on the result of
\cref{lem:epsg_lem} and proceed as follows:
\begin{itemize}
    \item In \cref{sssec:init_est_sp}, we show the functions $\hat{U}_i$ in the
    definition of $\phi^{(\tau)}$ \eqref{eq:fp_fin_def} (i.e., the coarse
    approximations of $U^*_i$) may be efficiently estimated from data.
    \item In \cref{ssec:macro_ints}, we show how to construct the
    aforementioned data-driven micro- and macro-intervals using the functions
    $\hat{U}_i$.
    \item Subsequently, in \cref{sssec:contr}, we analyze the contractivity
    properties of $\phi^{(\tau)}$ over each macro-interval, allowing application
    of \cref{thm:banach}. 
    \item In \cref{sssec:err_prop}, we show how errors incurred in early stages
    of the procedure are exponentially compounded for each new
    \emph{macro-interval} requiring careful control over the number of such
    intervals, $T$. 
    \item Finally, we bound $T$ and prove \cref{thm:sec_price_fin_samp} in
    \cref{sssec:iter_count}. 
\end{itemize}

\subsubsection{Approximate Estimation}
\label{sssec:init_est_sp}
\begin{algorithm}[h]
    \caption{Constructing macro- and micro-intervals for the fixed-point iteration.}
    \label{alg:est_u}
    \begin{algorithmic}[1]
    \Procedure{Coarse-Approximate-$U$}{$\{(Y_j, Z_j)\}_{j=1}^n, i$}
       \State \Return $x \mapsto \frac{1}{n} \sum_{j = 1}^n \frac{1}{1 - Y_j} \cdot \bm{1} \lbrb{Z_j = i, Y_j \leq x}$
    \EndProcedure
    \end{algorithmic}
\end{algorithm}

Here we describe the construction of $\hat{U}_i$ used in \eqref{eq:fp_fin_def}, ensuring the truncation range contains the true parameter values; i.e. we establish the following lemma:
\begin{lemma}
    \label{lem:u_appx}
    Let $\hat{U}_i: [0, 1 - \theta / 4] \to \mb{R}$ be monotonic functions defined as follows:
    \begin{equation*}
        \hat{U}_i (x) \coloneqq \frac{1}{n} \sum_{j = 1}^n \frac{1}{1 - Y_j} \cdot \bm{1} \lbrb{W_j = i, Y_j \leq x}.
    \end{equation*}
    Then, for all $x, y \in [0, 1 - (\theta / 4)]$ such that $\max \lbrb{U^*_i(x) - U^*_i (y), \hat{U}_i (x) - \hat{U}_i (y)} \geq \lprp{\frac{\alpha \nu}{8\eta}}^{32 k}$, we have:
    \[
        \frac{\alpha}{2} (U^*_i (x) - U^*_i (y)) \leq \hat{U}_i (x) - \hat{U}_i (y) \leq 2 \eta (U^*_i (x) - U^*_i (y)).
    \]
\end{lemma}
\begin{proof}
    Fix $x, y$ satisfying the required constraints.
    We have:
    \begin{align*}
        \E [\hat{U}^i(x) - \hat{U}^i(y)] 
        &= \E \left[\frac{1}{1 - Y} \cdot \bm{1} \lbrb{W = i, y < Y \leq x}\right] \\
        &= \int_{x}^y \frac{1}{(1 - z)} \cdot (1 - F_i (z)) \sum_{j \neq i} f_j (z) \prod_{k \neq j, i} F_k (z) dz,
    \end{align*}
    and hence (since $\alpha z \leq F_i(z) \leq \eta z$), we get:
    \begin{equation*}
        \alpha (U^*_i (x) - U^*_i (y)) \leq 
        \E [\hat{U}^i(x) - \hat{U}^i(y)] 
        \leq \eta (U^*_i (x) - U^*_i (y)).
    \end{equation*}
    Now, we will show that the estimate $U^i$ can be uniformly estimated for all $i, x, y$. For empirical analysis, we have by the integration by parts formula:
    \begin{align*}
        \hat{U}_i(x) - \hat{U}_i(y) &\coloneqq \frac{1}{n} \cdot \sum_{j = 1}^n \frac{1}{1 - Y_j} \cdot \bm{1} \lbrb{W_j = i, y < Y_j \leq x} 
        \\
        &= \lprp{\frac{\hat{G}_i (y)}{1 - y} - \frac{\hat{G}_i (x)}{1 - x}} - \int_{y}^x \frac{1}{(1 - z)^2} \hat{G}_i (z) dz.
    \end{align*}
    Similarly, we have for the population counterparts:
    \begin{equation*}
        \E [\hat{U}^i(x) - \hat{U}^i(y)] 
        = \lprp{\frac{G_i (y)}{1 - y} - \frac{G_i (x)}{1 - x}} - \int_{y}^x \frac{1}{(1 - z)^2} \cdot G_i (z) dz.
    \end{equation*}
    From the previous two displays (and conditioning on \cref{lem:epsg_lem}), we get:
    \begin{align*}
        \abs{\hat{U}_i (x) - \hat{U}_i (y) - \E [\hat{U}^i(x) - \hat{U}_i(y)]} 
            &\leq 2 \cdot \frac{\eps_g}{\theta/4} +  
                 \int_{y}^x \frac{\eps_g}{(\theta/4)^2} dz \\
            &\leq (\theta/2) \cdot \frac{16\eps_g}{\theta^2} +  
                 (x - y) \frac{16 \eps_g}{\theta^2} \leq \frac{32\eps_g}{\theta^2}.
    \end{align*}
    Hence, we get:
    \begin{equation*}
        \alpha (U^*_i (x) - U^*_i (y)) - \frac{32 \eps_g}{\theta^2} \leq \hat{U}_i (x) - \hat{U}_i (y) \leq \eta (U^*_i (x) - U^*_i (y)) + \frac{32 \eps_g}{\theta^2}
    \end{equation*}
    By \ref{eq:par_set}, this concludes the proof when $U^*_i (x) - U^*_i (y) \geq \lprp{\nicefrac{\alpha \nu}{8 \eta}}^{32k}$. In the alternative case (i.e when $\hat{U}_i (x) - \hat{U}_i (y) \geq \lprp{\nicefrac{\alpha \nu}{8 \eta}}^{32k}$), we again get from the above equation:
    \begin{gather*}
        \alpha (U^*_i (x) - U^*_i (y)) \leq \hat{U}_i (x) - \hat{U}_i (y) + \frac{32 \eps_g}{\theta^2} \leq 2 \cdot \lprp{\hat{U}_i (x) - \hat{U}_i (y)} \\
        \eta (U^*_i (x) - U^*_i (y)) \geq \hat{U}_i (x) - \hat{U}_i (y) - \frac{32 \eps_g}{\theta^2} \geq \frac{1}{2} \cdot \lprp{\hat{U}_i (x) - \hat{U}_i (y)}
    \end{gather*}
    completing the proof in this case as well.
\end{proof}

\subsubsection{Defining Macro- and Micro-Intervals}
\label{ssec:macro_ints}
We now describe our procedure for splitting up our estimation interval $[\nu,
1-\theta]$ into macro- and micro-intervals. Recall that the micro-intervals (of
fixed length $\delta$) define the intervals on which our approximation will be
piecewise-constant, while the macro-intervals (consisting of collections of
micro-intervals whose size is dynamically determined by data) define which
groups of micro-intervals are estimated simultaneously in each stage (i.e., by
the same fixed-point iteration \eqref{eq:fp_fin_def}).

We construct these macro- and micro-intervals according to a few different
considerations: 
\begin{itemize}
    \item The length of the micro-intervals $\delta$ must be small enough to
    bound the approximation error of our estimates (we set $\delta$ in
    \eqref{eq:par_set});
    \item We set the length of the macro-intervals (dynamically) to ensure that
    the discretized version of fixed-point iteration remains contractive
    \item At the same time, we must ensure that the total number of macro-intervals,
    $T$, does not grow too rapidly, as estimation errors incurred in earlier
    stages of the algorithm are exponentially amplified in later stages. 
\end{itemize}

We construct our micro- and macro-intervals in a recursive fashion, detailed in \cref{alg:macro_ints}. The algorithm assumes access to the coarse approximations $\hat{U}_i$ to $U^*_i$ for any $i \in [k]$---we established access to such approximations in the previous section (\cref{lem:u_appx}). In particular, starting from $\nu$ as the first endpoint, we set the macro-interval length based on a worst-case estimate of the resulting fixed-point contractivity, obtained using the coarse approximations. This quantity which plays a crucial role in our algorithm design and subsequent analysis is defined below:

\begin{equation*}
    \forall \tau \in [T], \ell \in \mathbb{N} \text{ s.t } x_{\tau, \ell} \leq 1 - \frac{\theta}{2}: \gamma^{(\tau)}_\ell \coloneqq  \max_{i \in [k]} 16 \cdot \lprp{\frac{\eta}{\alpha}}^6  \cdot \frac{1}{\hat{U}_i (x_{\tau - 1})} \sum_{m = 1}^{\ell} \frac{x_{\tau, m}}{(1 - x_{\tau, m})^2} \cdot \Delta^{(\tau)}_{i, m}.
\end{equation*}

\begin{algorithm}[H]
    \caption{Constructing macro- and micro-intervals for the fixed-point iteration.}
    \label{alg:macro_ints}
    \begin{algorithmic}[1]
    \Procedure{Construct-Macro-Intervals}{$\{(Y_j, Z_j)\}_{j=1}^n, \delta,
    \alpha, \eta, \nu, \theta$}
       \State $x_{0} \gets \nu,\ \tau \gets 1$
       \State Let $\hat{U}_i(\cdot) \gets $ \Call{Coarse-Approximate-$U$}{
           $\{(Y_j, Z_j)\}_{j=1}^n$, $i$} for all $i \in [k]$ \Comment{See \cref{lem:u_appx}}
       \While{$x_{\tau - 1} < 1 - \theta$}
        \State Number of micro-intervals in this macro-interval:
        \State $\ell^{(\tau)} \gets \max \lbrb{\ell \in \mb{N}: x_{\tau,
        \ell} \leq \min(2 \cdot x_{\tau - 1}, 1 - \theta / 2) \text{ and
        } \gamma^{(\tau)}_\ell \leq 1/4}.$
        \State $x_{\tau} \gets x_{\tau, \ell^{(\tau)}}, \tau \gets \tau + 1$
       \EndWhile
       \State $T \gets \tau - 1$
       \State \Return $\{x_{t}\}_{t \in [T]}$
    \EndProcedure
    \end{algorithmic}
\end{algorithm}

\subsubsection{Contractivity Analysis}
\label{sssec:contr}

We now establish contractivity of the mappings, $\phi^{(\tau)}$ in
\cref{lem:contr}, allowing us to apply Banach's fixed point theorem
(\cref{thm:banach}). In particular, we establish contractivity in the
infinity-norm; i.e, for some $\rho < 1$: 
\begin{equation*}
    \norm*{\phi^{(\tau)} (U) - \phi^{(\tau)} (U')}_\infty \leq \rho \norm{U - U'}_\infty \text{ where } \norm{M}_\infty = \max_{i, j} \abs{M_{i, j}}.
\end{equation*}
Denoting the Jacobian of $\phi^{(\tau)}$ by $J_{\phi^{(\tau)}} (\cdot)$, we consequently bound its $1$-norm defined below:
\begin{equation*}
    \norm*{J_{\phi^{(\tau)}} (U^{(\tau)})}_1 \coloneqq \max_{i, l} \sum_{\substack{j \in [k] \\ m \in [\ell^{(\tau)}]}} \lprp{J_{\phi^{(\tau)}} (U^{(\tau)})}_{(i, l), (j, m)} = \max_{i, l} \sum_{\substack{j \in [k] \\ m \in [\ell^{(\tau)}]}} \abs*{\frac{\partial \phi^{(\tau)}_{i, l} (U^{(\tau)})}{\partial U^{(\tau)}_{j, m}}}
\end{equation*}

Before proceeding further, we first establish that $\phi^{(\tau)}$ is a valid
fixed-point iteration, in the sense that it maps a compact set $S$ to itself:
\begin{lemma}
\label{lem:phi_monotone}
The mapping $\phi^{(\tau)}$ maps the set $S^{(\tau)}$ defined as follows
onto itself:
\begin{equation*}
    S^{(\tau)} \coloneqq \lbrb{U^{(\tau)} \in \mb{R}^{k \times \ell^{(\tau)}}: 
    \frac{1}{2\eta} \cdot \hat{U}_i (x_{\tau, l}) \leq U^{(\tau)}_{i, l} 
    \leq \frac{2}{\alpha} \cdot \hat{U}_i (x_{\tau, l}) \text{ and } U^{(\tau)}_{i, l} \leq U^{(\tau)}_{i, l + 1}}.
\end{equation*}
\end{lemma}
\begin{proof}
    The first constraint follows from \ref{eq:fp_fin_def} while the second
    follows from the fact that $\phi^{\tau}$ is a clipping of a monotonic
    function onto a monotonically growing range (\cref{lem:u_appx}).  
\end{proof}

Having established the existence of a set $S^{(\tau)}$ such that $\phi^{(\tau)}:
S^{(\tau)} \to S^{(\tau)}$, we can proceed to bounding the contractivity of the
fixed-point iteration itself:
\begin{lemma}
    \label{lem:contr}
    We have, for all $U^{(\tau)} \in S^{(\tau)}$,
    \begin{equation*}
        \norm*{J_{\phi^{(\tau)}} (U^{(\tau)})}_1 \leq \max_{i \in [k]} 16 \cdot \lprp{\frac{\eta}{\alpha}}^6  \cdot \frac{1}{\hat{U}_i (x_{\tau - 1})} \sum_{m = 1}^{\ell^{(\tau)}} \frac{x_{\tau, m}}{(1 - x_{\tau, m})^2} \cdot \Delta^{(\tau)}_{i, m}.
    \end{equation*}
\end{lemma}
\begin{proof}
    First, define the function $\Gamma^{(\tau)}_{i,l}$ as 
    \begin{align*}
        \Gamma_{i,l}^{(\tau)}(U^{(\tau)}) \coloneqq 
            \sum_{m = 0}^l \frac{1}{1 - H^{(\tau)}_{i, m}(U^{(\tau)})} \cdot \Delta^{(\tau)}_{i, m} + V^{(\tau)}_{i},
    \end{align*}
    so that  $\phi^{(\tau)}_{i,l}(U^{(\tau)}) = \text{clip}(
        \Gamma_{i,l}^{(\tau)},
        \frac{1}{2\eta} \cdot \hat{U}_i (x_{\tau, l}), 
        \frac{2}{\alpha} \cdot \hat{U}_i (x_{\tau, l}) )$.
    Direct calculation shows for any $i, j, l, m$,
    \begin{align}
        \label{eq:gamma_deriv_sp}
        \frac{\partial}{\partial U^{(\tau)}_{j, m}} \Gamma^{(\tau)}_{i,l}(U^{(\tau)}) 
        &=  \frac{1}{\lr{1 - H^{(\tau)}_{i, m}(U^{(\tau)})}^2} \cdot \Delta^{(\tau)}_{i, m}
            \cdot
            \frac{\partial H^{(\tau)}_{i,m}(U^{(\tau)})}{\partial U^{(\tau)}_{j, m}}.
    \end{align}
    It remains to compute the partial derivative ${\partial
    H^{(\tau)}_{i,m}(U^{(\tau)})} / {\partial U^{(\tau)}_{j, m}}$, which (from
    \eqref{eq:fp_fin_def}) corresponds to
    \begin{align}
        \label{eq:h_deriv_sp}
        \frac{\partial H^{(\tau)}_{i,m}(U^{(\tau)})}{\partial U^{(\tau)}_{j, m}} 
        &= \lr{\frac{1}{k-1} - \bm{1}\lbrb{i = j}} \frac{1}{U^{(\tau)}_{j,m}} H^{(\tau)}_{i,m}(U^{(\tau)}).
    \end{align}
    Now, combining \eqref{eq:gamma_deriv_sp} and \eqref{eq:h_deriv_sp} allows us
    to consider a single term in the Jacobian of $\phi^{(\tau)}$:
    \begin{align*}
        \abs*{\frac{\partial \phi^{(\tau)}_{i, l} (U^{(\tau)})}{\partial U^{(\tau)}_{j, m}}}
        \leq 
        \abs*{\frac{\partial \Gamma^{(\tau)}_{i, l} (U^{(\tau)})}{\partial U^{(\tau)}_{j, m}}}
        \leq 
        \abs*{\frac{1}{k-1} - \bm{1}\lbrb{i = j}}
        \cdot
        \abs*{\frac{H^{(\tau)}_{i,m}(U^{(\tau)})}{\lr{1 - H^{(\tau)}_{i, m}(U^{(\tau)})}^2} \cdot
        \Delta^{(\tau)}_{i, m}
            \cdot
            \frac{1}{U^{(\tau)}_{j,m}} 
        }.
    \end{align*}
    By construction (i.e., by the clipping in \eqref{eq:fp_fin_def}), 
    $H^{(\tau)}_{i, m}(U^{(\tau)}) \leq \min(\eta \cdot x_{\tau, m}, 1 - \alpha(1 - x_{\tau,m}))$, 
    and so:
    \begin{align*}
        \abs*{\frac{\partial \phi^{(\tau)}_{i, l} (U^{(\tau)})}{\partial U^{(\tau)}_{j, m}}}
        &\leq \frac{\eta}{\alpha^2} \cdot \abs*{\frac{1}{k - 1} - \bm{1} \lbrb{i = j}} \cdot \abs*{\frac{x_{\tau, m}}{(1 - x_{\tau, m})^2} \cdot \frac{1}{U^{(\tau)}_{j, m}} \cdot \Delta^{(\tau)}_{i, m}}.
    \end{align*}
    Now, note that $\smash{U^{(\tau)}_{j, m}}$ is either equal to
    $\smash{\hat{U}_j(x_{\tau, m})}$ (at initialization), or it is the output of
    $\phi^{(\tau)}_{j, m}$. As a result (again, by the
    clipping construction in \eqref{eq:fp_fin_def}), we must have (also using
    our approximation result from \cref{lem:u_appx})
    \[
        U^{(\tau)}_{j, m} \geq \frac{1}{2\eta} \hat{U}_j(x_{\tau, m})
                          \geq \frac{\alpha}{4\eta} U^*_j(x_{\tau, m}).
    \]
    We will also need the following Lemma relating the functions $U_i^*(\cdot)$
    across agents:
    \begin{restatable}{lemma}{uiujcomp}
        \label{lem:ui_uj_comp}
        Let $U^*_i(\cdot)$ be as defined in \eqref{eq:fixed_point}. Then, for
        all $i, j \in [k]$ and for any $x, y \in [0, 1]$ with $x > y$, we have
        that:
        \begin{equation*}
            \lprp{\frac{\alpha}{\eta}}^3 (U^*_j (x) - U^*_j (y)) \leq U^*_i (x) - U^*_i (y) \leq \lprp{\frac{\eta}{\alpha}}^3 (U^*_j (x) - U^*_j (y)).
        \end{equation*}
    \end{restatable}
    \begin{proof}
        Algebraic manipulation of \eqref{eq:fixed_point} and
        \cref{as:second_price_bdd}---full proof in \cref{ssec:misc_res_sec_price}.
    \end{proof}
    \noindent Applying both of these results to the Jacobian entry and applying
    \cref{lem:u_appx} again,
    \begin{align*}
        \abs*{\frac{\partial \phi^{(\tau)}_{i, l} (U^{(\tau)})}{\partial U^{(\tau)}_{j, m}}}
        &\leq 4 \cdot \frac{\eta}{\alpha^2} \cdot \lprp{\frac{\eta}{\alpha}}^4 \cdot \abs*{\frac{1}{k - 1} - \bm{1} \lbrb{i = j}} \cdot \abs*{\frac{x_{\tau, m}}{(1 - x_{\tau, m})^2} \cdot \frac{1}{U^*_i (x_{\tau, m})} \cdot \Delta^{(\tau)}_{i, m}} \\
        &\leq 8 \cdot \lprp{\frac{\eta}{\alpha}}^6 \cdot \abs*{\frac{1}{k - 1} - \bm{1} \lbrb{i = j}} \cdot \abs*{\frac{x_{\tau, m}}{(1 - x_{\tau, m})^2} \cdot \frac{1}{\hat{U}_i (x_{\tau - 1})} \cdot \Delta^{(\tau)}_{i, m}}.
    \end{align*}
    \noindent Finally, summing the previous equation over all $j, m$ yields:
    \begin{align*}
        \norm*{J_{\phi^{(\tau)}}}_{1} \leq \max_{i \in [k]} 16 \cdot \lprp{\frac{\eta}{\alpha}}^6  \cdot \frac{1}{\hat{U}_i (x_{\tau - 1})} \sum_{m = 1}^{\ell^{(\tau)}} \frac{x_{\tau, m}}{(1 - x_{\tau, m})^2} \cdot \Delta^{(\tau)}_{i, m}.
    \end{align*}
\end{proof}
As a direct consequence of this Lemma, we can show that $\phi^{(\tau)}(\cdot)$
is $\nicefrac{1}{4}$-contractive for our particular setting of $\delta$:
\begin{corollary}
    \label{cor:contractive}
    For any $\tau \in [T]$, let $\phi^{(\tau)}$ be as defined in
    \eqref{eq:fp_fin_def} and $\delta$ as set in \eqref{eq:par_set}, 
    \[
       \norm*{J_{\phi^{(\tau)}}}_{1} \leq \frac{1}{4}.
    \]
\end{corollary} 

\subsubsection{Error Propagation}
\label{sssec:err_prop}
One consequence of our stagewise estimation procedure is that error incurred
during the early stages {\em compounds} into the later stages. 
In this subsection, we establish a bound on the rate at which this error
grows---in particular, we show that the total error incurred over the course of
our entire estimation procedure is at most exponential in the number of
\emph{macro-intervals}, $T$.

We first introduce some notation. Let $\smash{\wt{U}^{(\tau)}}$ 
be estimates resulting from running the fixed-point iteration $\phi^{(\tau)}$
\eqref{eq:fp_fin_def} $L$ times---the results of the last subsection (in
particular, \cref{cor:contractive}) imply the existence of a unique fixed-point 
$\smash{\wt{U}^{(\tau)}}_{\text{fixed}}$ for which 
\begin{equation}
    \label{eq:fp_error}
    \norm*{
        \wt{U}^{(\tau)} - \wt{U}^{(\tau)}_{\text{fixed}}
    }_{\infty} \leq \frac{1}{4^L}.
\end{equation}
We will compare the estimates $\wt{U}^{(\tau)}$ to the ``ground-truth'' 
variables $\overline{U}^{(\tau)} \in \mb{R}^{k \times \ell^{(\tau)}}$, defined 
as:
\begin{equation*}
    \forall \tau \in [T], i \in [k], l \in \ell^{(\tau)}: \bar{U}^{(\tau)}_{i, l} \coloneqq U^*_i (x_{\tau, l}).
\end{equation*}
With this notation in hand, we can bound the total error incurred from our
stage-wise estimation algorithm, which we accomplish via induction:
\begin{lemma}
    \label{lem:err_prop}
    Let $T \in \mathbb{N}$ be the number of macro-intervals in our estimation
    algorithm. Then,
    \begin{equation*}
        \max_{\tau \in [T]}\ \norm*{\wt{U}^{(\tau)} - \bar{U}^{(\tau)}}_\infty \leq 2^{T} (2\eta \nu)^{k}.
    \end{equation*}
\end{lemma}
\begin{proof}
        We will prove the lemma by induction on the number of macro intervals. Concretely, we will prove the following claim via induction on $\tau$:
        \begin{equation*}
            \norm*{\wt{U}^{(\tau)} - \bar{U}^{(\tau)}} \leq 2^{\tau} (2\eta \nu)^{k} \tag{IND} \label{eq:ind_err}.
        \end{equation*}
        We first establish the base case:
        \begin{align*}
            \abs{\hat{G}_i (\nu) - U^*_i (\nu)} &\leq \eps_g + \abs{G_i (\nu) - U^*_i (\nu)} \\
            &= \eps_g + \int_{0}^\nu (1 - (1 - F_i(z))) \cdot \sum_{j \neq i} f_j (z) \prod_{m \neq i, j} F_m (z) dz \\
            &\leq \eps_g +  \eta \nu \int_{0}^\nu \sum_{j \neq i} f_j (z) \prod_{m \neq i, j} F_m (z) dz \\ 
            &\leq \eps_g + \eta \nu U^*_i (\nu) \leq (2\eta \nu)^{k}
        \end{align*}
        
        \noindent For the induction step, suppose \eqref{eq:ind_err} is true for
        all intervals up to $\tau - 1$. From \eqref{eq:fp_error}, 
        it suffices to bound the error between $U^{(\tau)}_{\mrm{fixed}}$ and $\bar{U}^{(\tau)}$:
        \begin{align*}
            \norm{\bar{U}^{(\tau)} - U^{(\tau)}_{\mrm{fixed}}}_\infty &\leq \norm{\bar{U}^{(\tau)} - \phi^{(\tau)}(\bar{U}^{(\tau)})}_\infty + \norm{\phi^{(\tau)}(\bar{U}^{(\tau)}) - U^{(\tau)}_{\mrm{fixed}}}_\infty \\
            &= \norm{\bar{U}^{(\tau)} - \phi^{(\tau)}(\bar{U}^{(\tau)})}_\infty + \norm{\phi^{(\tau)}(\bar{U}^{(\tau)}) - \phi^{(\tau)} (U^{(\tau)}_{\mrm{fixed}})}_\infty \\
            &\leq \norm{\bar{U}^{(\tau)} - \phi^{(\tau)}(\bar{U}^{(\tau)})}_\infty + {\norm{\bar{U}^{(\tau)} - U^{(\tau)}_{\mrm{fixed}}}_\infty}/{4} \\
            \norm{\bar{U}^{(\tau)} - U^{(\tau)}_{\mrm{fixed}}}_\infty 
            &\leq \frac{4}{3}\norm{\bar{U}^{(\tau)} - \phi^{(\tau)}(\bar{U}^{(\tau)})}_\infty \numberthis \label{eq:ust_baru_contr_bnd}. \\
        \end{align*}
        For the RHS, we have for fixed $i \in [k], \ell \in [\ell^{(\tau)}]$:
        \begin{align*}
            &\abs*{\bar{U}^{(\tau)}_{i, \ell} - (\phi^{(\tau)} (\bar{U}^{(\tau)}))_{i, \ell}} 
            = \abs*{\bar{U}_{i, \ell} - \sum_{l = 1}^\ell \frac{1}{1 - H^{(\tau)}_{i, l} (\bar{U}^{(\tau)})} \Delta^{(\tau)}_{i, l} - V_i^{(\tau)}} \\
            &= \abs*{\int_{x_{\tau - 1}}^{x_{\tau, \ell}} \frac{1}{1 - F_i (z)} \cdot g_i (z) dz + U^*_i (x_{\tau - 1}) - \sum_{l = 1}^\ell \frac{1}{1 - F_i(x_{\tau, l})} \Delta^{(\tau)}_{i, l} - V_i^{(\tau)}} \\
            &\leq \abs*{\int_{x_{\tau - 1}}^{x_{\tau, \ell}} \frac{1}{1 - F_i (z)} \cdot g_i (z) dz - \sum_{l = 1}^\ell \frac{1}{1 - F_i(x_{\tau, l})} \Delta^{(\tau)}_{i, l}} + \abs*{U^*_i (x_{\tau - 1}) - V_i^{(\tau)}} \\
            &\leq \abs*{\int_{x_{\tau - 1}}^{x_{\tau, \ell}} \frac{1}{1 - F_i (z)} \cdot g_i (z) dz - \sum_{l = 1}^\ell \frac{1}{1 - F_i(x_{\tau, l})} \Delta^{(\tau)}_{i, l}} + \norm*{\wt{U}^{(\tau - 1)} - \bar{U}^{(\tau - 1)}}_\infty. 
            \numberthis \label{eq:u_phiu_bnd}
        \end{align*}
    
        \noindent The second term in the above expression is bounded by our induction hypothesis---for the first term, we have:
        \begin{align*}
            &\!\!\!\!\!\!\!\!\!\!\!\!\!\!\!\!\!\!
            \abs*{\int_{x_{\tau - 1}}^{x_{\tau, \ell}} \frac{1}{1 - F_i (z)} \cdot g_i (z) dz - \sum_{l = 1}^\ell \frac{1}{1 - F_i(x_{\tau, l})} \Delta^{(\tau)}_{i, l}} \\
            &= \abs*{\sum_{l = 1}^\ell \int_{x_{\tau, l - 1}}^{x_{\tau, l}} \frac{1}{1 - F_i (z)} \cdot g_i (z) dz - \frac{1}{1 - F_i(x_{\tau, l})} (\hat{G}_i (x_{\tau, l}) - \hat{G}_i (x_{\tau, l - 1}))} \\
            &\leq \abs*{\sum_{l = 1}^\ell \int_{x_{\tau, l - 1}}^{x_{\tau, l}} \lprp{\frac{1}{1 - F_i (z)} - \frac{1}{1 - F_i (x_{\tau, l})}} \cdot g_i (z) dz} + \\
            &\qquad \qquad \abs*{\sum_{l = 1}^\ell \frac{1}{1 - F_i(x_{\tau, l})} (\hat{G}_i (x_{\tau, l}) - \hat{G}_i (x_{\tau, l - 1}) - (G_i (x_{\tau, l}) - G_i (x_{\tau, l - 1})))} \\
            &\leq \sum_{l = 1}^\ell \int_{x_{\tau, l - 1}}^{x_{\tau, l}} \frac{(F_i (x_{\tau, l}) - F_i (z))}{(1 - F_i (z))(1 - F_i (x_{\tau, l}))} \cdot g_i (z) dz + 
            \frac{8 \eps_g \ell^{(\tau)}}{\alpha \theta} \\
            &\leq \frac{4 \eta \delta}{(\alpha \theta)^2} \int_{x_{\tau, 0}}^{x_{\tau, \ell}} g_i (z) dz + \frac{8\eps_g \ell^{(\tau)}}{\alpha \theta} \\
            &\leq \frac{4 \eta \delta}{(\alpha \theta)^2} + \frac{8\eps_g \ell^{(\tau)}}{\alpha \theta}.
            \numberthis \label{eq:u_phiu_int_bnd}
        \end{align*}
        \cref{eq:fp_error,eq:ust_baru_contr_bnd,eq:u_phiu_bnd,eq:u_phiu_int_bnd} conclude the induction step with \eqref{eq:par_set} and \eqref{eq:ind_err}.
        
\end{proof}

\subsubsection{Bounding the Number of Macro Intervals}
\label{sssec:iter_count}
We now turn to the most technical component of our proof of
\cref{thm:sec_price_fin_samp}, namely bounding $T$ (the number of
macro-intervals in our algorithm). 
In \cref{lem:T_bnd}, we employ a potential function argument and track the
growth of $U^*_i$ at the end points of the macro-intervals $\{ x_{\tau} \}_{\tau
\in [T]}$ to argue that the total number of such intervals must be bounded.

\paragraph{Preliminaries.} 
Before proving \cref{lem:T_bnd}, we establish a few
preliminary results. Our first proves establishes a threshold beyond which the functions $U^*_i$ are lower bounded by a constant.
\begin{lemma}
    \label{lem:u_ub_lem}
    We have:
    \begin{equation*}
        \forall i \in [k], x \in \lsrs{1 - \frac{1}{4 \eta k}, 1}: U^*_i (x) \geq \frac{3}{4}.
    \end{equation*}
\end{lemma}
\begin{proof}
    Let $X_i \overset{iid}{\thicksim} F_i$. Then, by the union bound:
    \(
        \mb{P} \lbrb{\exists i: X_i \geq 1 - \frac{1}{4\eta k}} \leq \frac{1}{4},
    \)
    concluding the proof.
\end{proof}

Our second result is an upper bound on the densities of the ``densities'' $g_i (\cdot)$.
\begin{lemma}[Bounding the density of $g_i(\cdot)$]
    \label{lem:g_bnd}
    Under \cref{as:second_price_bdd}, we have that for all $i \in [k]$ and all
    $x \in [0, 1]$, the density $g_i$ satisfies:
    \begin{equation*}
        g_i(x) \leq \frac{\eta^2}{\alpha}.
    \end{equation*}
\end{lemma}
\begin{proof}
    Fix $x \in [0, 1]$ and let $i^* = \argmax_{i \in  [k]} F_i (x)$. Then, for any $i
    \in [k]$, we get:
    \begin{align*}
        g_i(x) &= (1 - F_i (x)) \sum_{j \neq i} f_j(x) \prod_{\ell \neq i, j} F_\ell (x) \leq \eta (1 - F_i (x)) \sum_{j \neq i} \prod_{\ell \neq i, j} F_\ell (x) \\
               &\leq \lprp{\frac{\eta^2}{\alpha}} \cdot (1 - F_{i^*} (x)) \sum_{j \neq i} \prod_{\ell \neq i, j} F_\ell (x) \leq \lprp{\frac{\eta^2}{\alpha}} \cdot (1 - F_{i^*} (x)) \sum_{j \neq i} \prod_{\ell \neq i, j} F_{i^*} (x) \\
               &= (k - 1) \lprp{\frac{\eta^2}{\alpha}} (1 - F_{i^*} (x)) (F_{i^*} (x))^{k - 2} \leq \lprp{\frac{\eta^2}{\alpha}} \cdot \lprp{1 - \frac{1}{k - 1}}^{k - 2} \leq \frac{\eta^2}{\alpha}
    \end{align*}
    where the first inequality is \cref{as:second_price_bdd}, and the second is $(k-1)(1-x)x^{k-2} \leq 1\ \forall\ k \geq 2$. 
\end{proof}

Our next three results prove coarse niceness properties on the macro-intervals constructed by \cref{alg:macro_ints}. The first proves that \cref{alg:macro_ints} constructs a finite number of intervals and hence, terminates in polynomial time from our settings of $\delta$ \eqref{eq:par_set}. 
\begin{lemma}[Macro-intervals are non-empty]
    Running \cref{alg:macro_ints} with our
    particular parameters \eqref{eq:par_set} yields a set of non-empty
    macro-intervals that cover the interval $[\nu, 1-\theta]$. That is,
    \begin{equation*}
        \ell^{(\tau)} > 0\ \forall\ \tau \in [T] \qquad 
        \text{ and } 
        \qquad 
        x_T \geq 1 - \theta.
    \end{equation*}
\end{lemma}
\begin{proof}
    First, from \cref{alg:macro_ints}, we have that $x_{\tau, 0} < 1 - \theta$
    (otherwise the algorithm terminates). 
    Furthermore, we have $x_{\tau, 0} + \delta < 2 x_{\tau, 0}$, since $x_{\tau,
    0} > \nu > \delta$  from our definition of $\delta$ and $\nu$
    \eqref{eq:par_set}. Similarly, $x_{\tau, 0} + \delta < 1 - \theta/2$ since
    $x_{\tau,0} < 1 - \theta$ and $\delta < \theta/2$ (again, by
    \eqref{eq:par_set}). 
    
    It only remains to show that $\gamma^{(\tau)}(1)$ as
    defined in Line 8 of \cref{alg:macro_ints} is less than $1/4$.
    From \cref{lem:u_appx,lem:g_bnd}, our bounds on $\eps_g, \theta, \delta$
    (\ref{eq:par_set}) and the fact that $U^*_i (\nu) \geq (\alpha \nu)^{k - 1}$:
    \begin{align*}
        \gamma^{(\tau)}_1 &\coloneqq \max_{i \in [k]} 16 \cdot \lprp{\frac{\eta}{\alpha}}^6  \cdot \frac{1}{\hat{U}_i (x_{\tau - 1})} \cdot \frac{x_{\tau - 1} + \delta}{(1 - x_{\tau - 1} - \delta)^2} \cdot \Delta^{(\tau)}_{i, m} \\
        &\leq \max_{i \in [k]} 16 \cdot \lprp{\frac{\eta}{\alpha}}^6  \cdot \frac{1}{\hat{U}_i (x_{\tau - 1})} \cdot \frac{x_{\tau - 1} + \delta}{(1 - x_{\tau - 1} - \delta)^2} \cdot (\hat{G}_i (x_{\tau - 1} + \delta) - \hat{G}_i (x_{\tau - 1})) \\
        &\leq \max_{i \in [k]} 16 \cdot \lprp{\frac{\eta}{\alpha}}^6  \cdot \frac{1}{\hat{U}_i (x_{\tau - 1})} \cdot \frac{x_{\tau - 1} + \delta}{(1 - x_{\tau - 1} - \delta)^2} \cdot \lprp{\frac{\eta^2 \delta}{\alpha} + 2 \eps_g} \\
        &\leq \max_{i \in [k]} 16 \cdot \lprp{\frac{\eta}{\alpha}}^6  \cdot \frac{2}{\alpha} \cdot \frac{1}{U^*_i (\nu)} \cdot \frac{1 + \delta}{(\theta / 2 - \delta)^2} \cdot \lprp{\frac{\eta^2 \delta}{\alpha} + 2 \eps_g} < \frac{1}{4}
    \end{align*}
    establishing the first claim. Note that the previous argument also
    establishes the second claim as if $x_{T} < 1 - \theta$, a new
    macro-interval exists.  
\end{proof}

The next lemma lower bounds the value of our contractivity proxy, $\gamma^{(\tau)}_{\ell^{(\tau)}}$, for intervals, $[x_{\tau - 1}, x_\tau]$ where $x_\tau$ is not too large either with respect to $x_{\tau - 1}$ or the upper threshold $1 - \theta / 2$. We will subsequently establish in the proof of \cref{lem:T_bnd} that most intervals satisfy this condition and such intervals guarantee substantial growth in the potential functions employed in \cref{lem:T_bnd}.
\begin{lemma}
    \label{lem:gamma_lt_lb}
    Consider the context of our interval-construction algorithm 
    (\cref{alg:macro_ints}). We have:
    \begin{equation*}
        \forall \tau \in [T] \text{ s.t } x_{\tau - 1, \ell^{(\tau)} + 1} \leq \min(2 x_{\tau 
        - 1}, 1-\theta/2): \gamma^{(\tau)}_{\ell^{(\tau)}} \geq \frac{1}{8}.
    \end{equation*}
\end{lemma}
\begin{proof}
    We have:
    \begin{align*}
        \gamma_{\ell^{(\tau)}}^{(\tau)} &= \max_{i \in [k]} 16 \cdot \lprp{\frac{\eta}{\alpha}}^6  \cdot \frac{1}{\hat{U}_i (x_{\tau - 1})} \sum_{m = 1}^{\ell^{(\tau)}} \frac{x_{\tau, m}}{(1 - x_{\tau, m})^2} \cdot \Delta^{(\tau)}_{i, m} \\
        &= \max_{i \in [k]} 16 \cdot \lprp{\frac{\eta}{\alpha}}^6  \cdot \frac{1}{\hat{U}_i (x_{\tau - 1})} \lprp{\sum_{m = 1}^{\ell^{(\tau)} + 1} \frac{x_{\tau, m}}{(1 - x_{\tau, m})^2} \Delta^{(\tau)}_{i, m} - \frac{x_{\tau, \ell^{(\tau)} + 1}}{(1 - x_{\tau, \ell^{(\tau)} + 1})^2} \Delta^{(\tau)}_{i, \ell^{(\tau)} + 1}} \\
        &\geq \gamma^{(\tau)}_{\ell^{(\tau)} + 1} - 16 \cdot \lprp{\frac{\eta}{\alpha}}^6 \cdot \frac{1}{\hat{U}_i (x_{\tau - 1})} \cdot \frac{1 - \theta / 4}{(\theta / 4)^2} \cdot \lprp{\frac{\eta^2 \delta}{\alpha} + 2 \eps_g} \\
        &\geq \gamma^{(\tau)}_{\ell^{(\tau)} + 1} - 16 \cdot \lprp{\frac{\eta}{\alpha}}^6 \cdot \frac{2}{\alpha} \cdot \frac{1}{U^*_i (x_{\tau - 1})} \cdot \frac{1 - \theta / 4}{(\theta / 4)^2} \cdot \lprp{\frac{\eta^2 \delta}{\alpha} + 2 \eps_g} \\
        &\geq \gamma^{(\tau)}_{\ell^{(\tau)} + 1} - 16 \cdot \lprp{\frac{\eta}{\alpha}}^6 \cdot \frac{2}{\alpha} \cdot \frac{1}{U^*_i (\nu)} \cdot \frac{1 - \theta / 4}{(\theta / 4)^2} \cdot \lprp{\frac{\eta^2 \delta}{\alpha} + 2 \eps_g} \geq \frac{1}{8}
    \end{align*}
    where the first inequality follows from \cref{lem:g_bnd} and noting $\delta < \theta / 4$ and $x_T \leq 1 - \theta / 2$ and the next two from \cref{lem:u_appx}, our bounds on $\eps_g, \delta, \theta$ (\ref{eq:par_set}) and observing $U^*_i (\nu) \geq (\alpha \nu)^{k - 1}$.
\end{proof}

The final preliminary result guarantees that the conclusion of \cref{lem:u_appx} holds for all intervals $[x_{\tau - 1}, x_\tau]$ constructed by \cref{alg:macro_ints} and will be used to guarantee large potential growth for intervals satisfying the conclusion of \cref{lem:gamma_lt_lb}.

\begin{lemma}
    \label{lem:ust_diff}
    Consider the context of our interval-construction algorithm (\cref{alg:macro_ints}). We have:
    \begin{gather*}
        \forall \tau \in [T], i \in [k]: U^*_i (x_\tau) - U^*_i (x_{\tau - 1}) \geq \frac{1}{2048} \cdot \lprp{\frac{\alpha}{\eta}}^{10} \cdot \theta \cdot \lprp{\frac{\alpha \nu}{2}}^{k - 1}.
    \end{gather*}
\end{lemma} 
\begin{proof}
    To start, suppose $\tau \in [T]$. We prove the lemma in three cases:
    
    \paragraph{Case 1:} $x_{\tau - 1, \ell^{(\tau)} + 1} \leq \min \lprp{2 x_{\tau - 1}, 1 - \theta / 2}$. \cref{lem:gamma_lt_lb} implies $\gamma^{(\tau)}_{\ell^{(\tau)}} \geq 1 / 8$. We have for some $i$ from the facts that $x_{\tau, m} < 1 - \theta / 2$ for all $m \in  [\ell^{(\tau)}]$ and $\hat{U}_{i} (x_{\tau - 1}) \geq \hat{G}_i (\nu) \geq (1 - \eta \nu) (\alpha \nu)^{k - 1} - \eps_g$ along with our parameter settings \eqref{eq:par_set}:
    \vspace{-5pt}
    \begin{equation*}
        64 \lprp{\frac{\eta}{\alpha}}^6  \frac{1}{(\alpha \nu)^{k - 1}} \frac{1}{\theta} \sum_{m = 1}^{\ell^{(\tau)}} \frac{1}{(1 - x_{\tau, m})} \Delta^{(\tau)}_{i, m} \geq 16 \lprp{\frac{\eta}{\alpha}}^6  \frac{1}{\hat{U}_i (x_{\tau - 1})} \sum_{m = 1}^{\ell^{(\tau)}} \frac{x_{\tau, m}}{(1 - x_{\tau, m})^2} \Delta^{(\tau)}_{i, m} \geq \frac{1}{8}
    \end{equation*}
    Re-arranging the above and applying \cref{lem:uh_disc_apx,lem:gamma_lt_lb,lem:ui_uj_comp} yields for all $j \in [k]$:
    \begin{align*}
        2\lprp{\frac{\eta}{\alpha}}^4 (U^*_j (x_{\tau}) - U^*_j (x_{\tau - 1})) &\geq 2\eta (U^*_i (x_{\tau}) - U^*_i (x_{\tau - 1})) \\
        &\geq \hat{U}_i (x_{\tau}) - \hat{U}_i (x_{\tau - 1}) \geq \frac{1}{1024} \cdot \lprp{\frac{\alpha}{\eta}}^6 \cdot \theta \cdot (\alpha \nu)^{k - 1}
    \end{align*}
    proving the claim in this case. 
    
    \paragraph{Case 2:} $x_{\tau - 1, \ell^{(\tau)} + 1} > 2 x_{\tau - 1}$. As $x_{\tau - 1} \geq \nu$, $x_{\tau} \geq \frac{3}{2} x_{\tau - 1}$ from our setting of $\delta$ and the claim follows:
    \begin{equation*}
        U^*_i (x_{\tau}) - U^*_i (x_{\tau - 1}) = \int_{x_{\tau - 1}}^{x_{\tau}} \sum_{j \neq i} f_j (z) \prod_{m \neq i, j} F_m (z) dz \geq (k - 1) \alpha \int_{0}^{\frac{\nu}{2}} (\alpha z)^{k - 2} dz \geq \frac{(\alpha \nu)^{k - 1}}{2^{k - 1}}.
    \end{equation*}
    
    \paragraph{Case 3:} $x_{\tau - 1, \ell^{(\tau)} + 1} > 1 - \theta / 2$. Note that \cref{lem:gamma_lt_lb,alg:macro_ints} imply $\tau = T$, $x_{\tau - 1} < 1 - \theta$ and $x_{\tau} + \delta \geq 1 - \theta / 2$. We obtain from our choice of $\delta$ that $x_\tau \geq 1 - (3 \theta) / 4$. Furthermore, since $x_{\tau - 1} < 1 - \theta$, we have from \cref{lem:u_ub_lem} and the constraint on $\eps$ in \cref{thm:sec_price_fin_samp}:
    \begin{align*}
        U^*_i (x_\tau) - U^*_i (x_{\tau - 1}) &\geq U^*_i \lprp{1 - \frac{3 \theta}{4}} - U^*_i (1 - \theta) = \int_{1 - \theta}^{1 - \frac{3 \theta}{4}} \sum_{j \neq i} f_j (z) \prod_{m \neq i, j} F_m (z) dz \\
        &\geq \alpha \int_{1 - \theta}^{1 - \frac{3 \theta}{4}} \sum_{j \neq i} \prod_{m \neq i, j} F_m (z) dz \geq \alpha \int_{1 - \theta}^{1 - \frac{3 \theta}{4}} (k - 1) U^*_i (1 - \theta) dz \geq \frac{\alpha\theta}{8}.
    \end{align*}
    concluding the proof of the lemma.
\end{proof}

A key observation behind our proof is that when $x_{\tau,0}$ is smaller than an 
appropriately chosen constant $c_{\eta, \alpha}$, the rate of growth of 
$U^*_i$ is {\em doubly exponential}, in that 
$$U^*_i (x_{\tau+1, 0}) \geq U^*_i (x_{\tau, 0})^{1 - \frac{1}{2 (k - 1)}}.$$
This rate of growth allows us to bound $T$
Hence, the rate of growth is \emph{doubly
exponential} before $c_{\eta, \theta}$ allowing the bound on $T \approx \log_2
(1 / \nu) + k \log \log (1 / \nu)$ while the initial error scales is at most
$\nu^{k}$. With \cref{lem:err_prop}, a simple post-processing step proves
\cref{thm:sec_price_fin_samp}. Auxiliary results needed in the proof are
included in \cref{ssec:misc_res_sec_price}. 

\begin{lemma}
    \label{lem:T_bnd}
    We have:
    \begin{equation*}
        T \leq 2(k - 1) \log \log (1 / (\alpha \nu)) + \log_2 (1 / \nu) + 2^{20} \lprp{\frac{\eta}{\alpha}}^{14} \lprp{k^2 \log (2\eta / \alpha) + k \log_2 (2 / \theta)}.
    \end{equation*}
\end{lemma}
\begin{proof}
    To bound $T$, we break $[\nu, 1 - \theta]$ into three segments and handle each separately:
\begin{equation*}
    I_1: \lsrs{\nu, \lprp{\frac{\alpha}{2\eta}}^{32}}, 
    I_2: \lsrs{\lprp{\frac{\alpha}{2\eta}}^{32}, 1 - \frac{1}{4\eta k}}, 
    \text{ and } 
    I_3: \lsrs{1 - \frac{1}{4\eta k}, 1 - \frac{\theta}{2}}. 
\end{equation*}

Before we proceed, we prove a simple claim allowing us to restrict ourselves to intervals, $[x_{\tau - 1}, x_{\tau}]$ such that $\gamma^{\tau}_{\ell^{(\tau)}}$ is large ($> 1 / 8$).
\begin{claim}
    \label{clm:few_bad_ints}
    We have:
    \begin{equation*}
        \abs{S} \leq \log_2 (4 / \nu) \text{ where } S \coloneqq \lbrb{\tau: \gamma^\tau_{\ell^{(\tau)}} \geq \frac{1}{8}}.
    \end{equation*}
\end{claim}
\begin{proof}
    Consider $\tau \in S$. Then, we have from \cref{lem:gamma_lt_lb} that either $x_{\tau} \geq 1 - \theta$ ($\tau = T$ as \cref{alg:macro_ints} now terminates) or $x_{\tau, \ell^{(\tau)} + 1} \geq 2 x_{\tau - 1}$. In the second case, we get:
    \begin{equation*}
        x_{\tau, \ell^{(\tau)} + 1} = x_{\tau} + \delta \geq 2 x_{\tau - 1} \implies x_{\tau} \geq \lprp{2 - \frac{\delta}{x_{\tau - 1}}} \cdot x_{\tau - 1} \geq x_{\tau} \geq \lprp{2 - \frac{\delta}{\nu}} \cdot x_{\tau - 1}.
    \end{equation*}
    Letting $T_S = \abs{S \setminus \{T\}}$, we get by iterating the above inequality and noting that $x_\tau$ are monotonic:
    \begin{equation*}
        1 \geq \lprp{2 - \frac{\delta}{\nu}}^{T_S} \cdot \nu
    \end{equation*}
    and hence, \ref{eq:par_set} yields:
    \begin{equation*}
        T_S \leq \frac{1}{\log_2 (2 - \delta / \nu)} \cdot \log_2 (1 / \nu) \leq \frac{1}{1 - \delta / \nu} \cdot \log_2 (1 / \nu) \leq \lprp{1 + 2 \cdot \frac{\delta}{\nu}} \cdot \log_2 (1 / \nu) \leq \log_2 (2 / \nu).
    \end{equation*}
    Noting that $\abs{S} \leq T_S + 1$ concludes the proof of the claim.
\end{proof}

\paragraph{Case 1:} $I_1$. We restrict ourselves to the intervals $[x_{\tau - 1}, x_\tau] \subset I_1$ where $\gamma^{(\tau)}_{\ell^{(\tau)}} \geq 1/8$. From the definition of $\gamma^{(\tau)}_{\ell^{(\tau)}}$, there exists some $i \in [k]$ such that: 
\begin{gather*}
    16 \cdot \lprp{\frac{\eta}{\alpha}}^6  \cdot \frac{1}{\hat{U}_i (x_{\tau - 1})} \sum_{m = 1}^{\ell^{(\tau)}} \frac{x_{\tau, m}}{(1 - x_{\tau, m})^2} \cdot \Delta^{(\tau)}_{i, m} \geq \frac{1}{8} \\
    U_i^* (x_{\tau - 1}) \geq (\alpha x_{\tau - 1})^{k - 1} \implies x_{\tau - 1} \leq \frac{U_i^*(x_{\tau - 1})^{1 / (k - 1)}}{\alpha}
\end{gather*}
and as a result, we obtain from \cref{lem:uh_disc_apx,lem:u_appx,lem:ust_diff}:
\begin{align*}
\frac{1}{8} &\leq 16 \cdot \lprp{\frac{\eta}{\alpha}}^6  \cdot \frac{1}{\hat{U}_i (x_{\tau - 1})} \cdot \frac{2 x_{\tau - 1}}{(1 - x_\tau)} \sum_{m = 1}^{\ell^{(\tau)}} \frac{1}{(1 - x_{\tau, m})} \cdot \Delta^{(\tau)}_{i, m} \\
&\leq 32 \cdot \lprp{\frac{\eta}{\alpha}}^6  \cdot \frac{1}{\hat{U}_i (x_{\tau - 1})} \cdot \frac{1}{\alpha} \cdot \frac{U^*_i (x_{\tau - 1})^{1 / (k - 1)}}{(1 - x_\tau)} \sum_{m = 1}^{\ell^{(\tau)}} \frac{1}{(1 - x_{\tau, m})} \cdot \Delta^{(\tau)}_{i, m} \\
&\leq 512 \cdot \lprp{\frac{\eta}{\alpha}}^8 \cdot \frac{1}{U^*_i (x_{\tau - 1})^{(k - 2) / (k - 1)}} \cdot \lprp{U^*_i (x_\tau) - U^*_i (x_{\tau - 1})}.
\end{align*}

Re-arranging the above, two applications of \cref{lem:ui_uj_comp} with the fact $U^*_j (x) \leq (\eta x)^{k - 1}$:
\begin{equation*}
    j \in [k]: U^*_j (x_\tau) \geq \lprp{\frac{\alpha}{2\eta}}^{14} \cdot U^*_j (x_{\tau - 1})^{(k - 2) / (k - 1)} \geq U_j^* (x_{\tau - 1})^{1 - \frac{1}{2(k - 1)}}.
\end{equation*}
Defining $S_1 \coloneqq \{\tau: [x_{\tau - 1}, x_\tau] \subset I_1 \text{ and } \gamma^{(\tau)}_{\ell^{(\tau)}} \geq 1/8\}$, $T_1 \coloneqq \abs{S_1}$ and $\tau_1^* \coloneqq \max S_1$, we get by a recursive application of the above inequality for all $j \in [k]$:
\begin{equation*}
e^{-(k - 1)} \geq U^*_j (x_{\tau_1^*}) \geq (U^*_j (x_0))^{{\lprp{1 - \frac{1}{2(k - 1)}}}^{T_1}} \geq (\alpha \nu)^{(k - 1) {\lprp{1 - \frac{1}{2(k - 1)}}}^{T_1}}.
\end{equation*}
Iteratively taking logs,
\begin{equation*}
    \exp \lbrb{- \frac{T_1}{2 (k - 1)}} \log (\alpha \nu) \geq 1 \implies T_1 \leq 2 (k - 1) \log \log (1 / (\alpha \nu)).
\end{equation*}

\paragraph{Case 2:} $I_2$. As before, we restrict to intervals $[x_{\tau - 1}, x_\tau] \subset I_2$ with $\gamma^{(\tau)}_{\ell^{(\tau)}} \geq 1 / 8$. Noting that $1 - x_{\tau} \geq 1 / (4\eta k)$, we have from \cref{lem:uh_disc_apx,lem:u_appx,lem:ust_diff} for some $i \in [k]$:
\begin{align*}
\frac{1}{8} &\leq 16 \cdot \lprp{\frac{\eta}{\alpha}}^6  \cdot \frac{1}{\hat{U}_i (x_{\tau - 1})} \cdot \frac{2 x_{\tau - 1}}{(1 - x_\tau)} \sum_{m = 1}^{\ell^{(\tau)}} \frac{1}{(1 - x_{\tau, m})} \cdot \Delta^{(\tau)}_{i, m} \\
&\leq 16 \cdot \lprp{\frac{\eta}{\alpha}}^6  \cdot \frac{1}{\hat{U}_i (x_{\tau - 1})} \cdot 8 \eta k \cdot \sum_{m = 1}^{\ell^{(\tau)}} \frac{1}{(1 - x_{\tau, m})} \cdot \Delta^{(\tau)}_{i, m} \\
&\leq 1024 k \cdot \lprp{\frac{\eta}{\alpha}}^8 \cdot \frac{1}{U^*_i (x_{\tau - 1})} \cdot \lprp{U^*_i (x_\tau) - U^*_i (x_{\tau - 1})}.
\end{align*}
Re-arranging the above inequality and two applications of \cref{lem:ui_uj_comp} yield:
\begin{equation*}
\forall j \in [k]: U^*_j (x_\tau) - U^*_j (x_{\tau - 1}) \geq \frac{1}{8192 k} \cdot \lprp{\frac{\alpha}{\eta}}^{11} \cdot U^*_i (x_{\tau - 1}) \geq \frac{1}{8192 k} \cdot \lprp{\frac{\alpha}{\eta}}^{14} \cdot U^*_j (x_{\tau - 1}).
\end{equation*}
Define $S_2 \coloneqq \{\tau: [x_{\tau - 1}, x_\tau] \subset I_2 \text{ and } \gamma_{\ell^{(\tau)}}^{(\tau)} \geq 1 / 8\}, T_2 \coloneqq \abs{S_2}$ and $\tau^*_2 \coloneqq \max S_2$ as before. As $x_{\tau - 1} \geq (\eta / 2\alpha)^{32}$ for all $\tau \in S_2$, recursively applying the above inequality yields:
\begin{equation*}
1 \geq U^*_j (x_{\tau^*_2}) \geq \lprp{1 + \frac{1}{8192k} \cdot \lprp{\frac{\alpha}{\eta}}^{14}}^{T_2} U^*_j \lprp{\lprp{\frac{\alpha}{2\eta}}^{32}}.
\end{equation*}
Again, noting $U^*_j(x) \geq (\alpha x)^{k - 1}$ and taking logs, the current case follows:
\begin{equation*}
1 \geq \exp \lbrb{\frac{T_2}{16384k} \cdot \lprp{\frac{\alpha}{\eta}}^{14}} \cdot \lprp{\frac{\alpha}{2 \eta}}^{64 (k - 1)} \implies T_2 \leq 2^{20} k^2 \lprp{\frac{\eta}{\alpha}}^{14} \log (2\eta / \alpha).
\end{equation*}

\paragraph{Case 3:} $I_3$. We start by subdividing $I_3$ into $r$ subintervals $\lbrb{I_{3, p} \coloneqq \lsrs{1 - 2^{p} \theta, 1 - 2^{p - 1} \theta}}_{p = 0}^r$. Note $r \leq \log_2 (2 / \theta) + 1$. We now bound the number of intervals in each of these sub-intervals and restrict ourselves to intervals $[x_{\tau - 1}, x_{\tau}] \subset I_{3, p}$ for some $p$ with $\tau < T$. Note this excludes at most $2(r + 1)$ intervals. For one such interval $[x_{\tau - 1}, x_\tau] \in I_{3, p}$, note that $\gamma_{\ell^{(\tau)}}^{(\tau)} \geq 1 / 8$ from \cref{lem:gamma_lt_lb} and the definition of $I_3$. Similarly, we have for some $i$ using \cref{lem:ust_diff,lem:u_ub_lem,lem:u_appx,lem:uh_disc_apx}:
\begin{align*}
\frac{1}{8} &\leq 16 \cdot \lprp{\frac{\eta}{\alpha}}^6  \cdot \frac{1}{\hat{U}_i (x_{\tau - 1})} \sum_{m = 1}^{\ell^{(\tau)}} \frac{x_{\tau, m}}{(1 - x_{\tau, m})^2} \cdot \Delta^{(\tau)}_{i, m} \\
&\leq 16 \cdot \lprp{\frac{\eta}{\alpha}}^6  \cdot \frac{1}{\hat{U}_i (x_{\tau - 1})} \sum_{m = 1}^{\ell^{(\tau)}} \frac{1}{(1 - x_{\tau, m})(2^{p - 1} \theta)} \cdot \Delta^{(\tau)}_{i, m} \\
&\leq \frac{256}{2^{p - 1} \theta} \cdot \lprp{\frac{\eta}{\alpha}}^7 \cdot \lprp{U^*_i (x_\tau) - U^*_i (x_{\tau - 1})}
\end{align*}
which yields:
\begin{equation*}
U^*_i (x_\tau) - U^*_i (x_{\tau - 1}) \geq \frac{2^{p - 1} \theta}{2048} \cdot \lprp{\frac{\alpha}{\eta}}^7 \implies \forall j \in [k]: U^*_j (x_\tau) - U^*_j (x_{\tau - 1}) \geq \frac{2^{p - 1} \theta}{2048} \cdot \lprp{\frac{\alpha}{\eta}}^{10}.
\end{equation*}
Again, defining $S_{3, p} = \{\tau < T: [x_{\tau - 1}, x_\tau] \subset I_{3, p} \}, T_{3, p} \coloneqq \abs{S_{3, p}}$, we have from the above:
\begin{align*}
T_{3, p} &\leq \frac{2048}{2^{p - 1} \theta} \cdot \lprp{\frac{\eta}{\alpha}}^{10} \lprp{U^*_i \lprp{1 - 2^{p - 1} \theta} - U^*_i \lprp{1 - 2^{p} \theta}} \\
&= \frac{2048}{2^{p - 1} \theta} \cdot \lprp{\frac{\eta}{\alpha}}^{10} \cdot \int_{1 - 2^p \theta}^{1 - 2^{p - 1} \theta} \sum_{j \neq i} f_j (x) \prod_{q \neq i,j} F_q(x) dx \\
&\leq \frac{2048}{2^{p - 1} \theta} \cdot \lprp{\frac{\eta}{\alpha}}^{10} \cdot (\eta k \cdot 2^{p - 1} \theta) \leq 2048 k \cdot \lprp{\frac{\eta}{\alpha}}^{11}.
\end{align*}
Summing up over $p$, we get that:
\begin{equation*}
    T_3 \coloneqq \sum_{p = 0}^r T_{r, p} \leq 4096 k \cdot \lprp{\frac{\eta}{\alpha}}^{11} \log_2 (2 / \theta).
\end{equation*}
Finally, \cref{clm:few_bad_ints}, accounting for intervals $[x_{\tau - 1}, x_\tau] \subsetneq I_j$ for all $j \in [3]$ and the previous three cases establish the lemma as:
\begin{equation*}
    T \leq T_1 + T_2 + T_{3} + 2(r + 1) + \log_2 (4 / \nu) + 2.
\end{equation*}
\end{proof}

\subsubsection{Completing the Proof of \cref{thm:sec_price_fin_samp}}
To complete the proof of \cref{thm:sec_price_fin_samp}, we have from
\cref{lem:err_prop,lem:T_bnd} and our setting of the parameter, $\nu$, that
\begin{equation*}
    \norm{\bar{U} - \wt{U}}_\infty \leq 2^T (2 \eta \nu)^k \leq (2 \eta \nu)^{k / 8} \leq \lprp{{\alpha \cdot \eps/32 \eta}}^{2k}.
\end{equation*}

We now recover estimates of $F_i$ from estimates of $U^*_i$. Note, when $x \leq \theta$, $F_i (x) \leq \eps / 16$ and hence, $0$ is suitable in this range. Likewise, when $x \geq 1 - \theta$, $F_i (x) \geq 1 - \eps / 16$ and $1$ is correspondingly accurate. For the final case, assume $\theta \leq x \leq 1 - \theta$. We will first estimate $F_i$ on the grid points, $x_{\tau, l}$. Suppose now that $x = x_{\tau, l}$ for some $\tau, l$. We have:
\begin{equation*}
    U^*_i (x) \geq \int_{0}^x \sum_{j \neq i} f_j (z) \prod_{m \neq i, j} F_m (z) dz \geq (k - 1) \alpha^{k - 1} \int_0^x z^{k - 2} dz \geq \lprp{\frac{\alpha \eps}{16}}^{k - 1}.
\end{equation*}
And, as a consequence, we get:
\(
    \lprp{1 - \frac{\eps}{16}} \leq \frac{U^*_i(x)}{\wt{U}^{(\tau)}_{i, l}} \leq \lprp{1 + \frac{\eps}{16}}.
\)
Defining our estimate:

\begin{center}
\(
    \hat{F}_i (x) \coloneqq \prod_{j \neq i} (\wt{U}^{(\tau)}_{j, l})^{1 / (k - 1)}
    /
    \lr{(U^{(\tau)}_{i, l})^{(k - 2) / (k - 1)}}
\)
\end{center}

\noindent we get by noting that $F_i (x) = \frac{\prod_{j \neq i} (U^*_j (x))^{1 / (k - 1)}}{(U^*_i (x))^{(k - 2) / (k - 1)}}$:
\begin{align*}
    \hat{F}_i(x) \leq F_i (x) \cdot \lprp{1 + \frac{\eps}{16}} \cdot \lprp{1 - \frac{\eps}{16}}^{-1} \leq \lprp{1 + \frac{\eps}{4}} \cdot F_i (x) \\
    \hat{F}_i(x) \geq F_i (x) \cdot \lprp{1 - \frac{\eps}{16}} \cdot \lprp{1 + \frac{\eps}{16}}^{-1} \geq \lprp{1 - \frac{\eps}{4}} \cdot F_i (x).
\end{align*}
Finally, for any $\theta \leq x \leq 1 - \theta$, there exists $x_{\tau, l}$ such that $\abs{x - x_{\tau, l}} \leq \delta$. And we have:
\begin{equation*}
    \frac{\abs{F_i (x) - \hat{F}_i (x_{\tau, l})}}{F_i (x)} \leq \frac{\abs{F_i (x) - F_i (x_{\tau, l})} + \abs{F_i (x_{\tau, l}) - \hat{F}_i (x_{\tau, l})}}{F_i(x_{\tau, l}) - \abs{F_i(x_{\tau, l}) - F_i (x)}} \leq \frac{\delta \eta + (\eps / 4) F_i (x_{\tau, l})}{F_i (x_{\tau, l}) - \delta \eta} \leq \frac{\eps}{2}
\end{equation*}
from our setting of $\delta$ and $\theta$. This concludes the proof of the theorem.
\qed
\subsection{Estimation from Partial Observations}
\label{sec:sp_bid_insert}
Finally, we explore a second-price analogue of the ``partial observability'' 
setting (introduced by \citet{blum2015learning}) that we studied in the context 
of first-price auctions in Section \ref{sec:1stPrice:additionalBidder}. 
In particular, in this setting we observe {\em the winner} of each auction
and a binary indicator of whether the reserve price was triggered, but
not the price that the winner pays for the auctioned good.
On the other hand, as in \citep{blum2015learning}, 
in this setting the econometrician is given the ability to 
set the reserve price (or equivalently, insert bids into the auction). 
We formally define the 
setting below:
\begin{definition}[Partial Observation Data -- Second-price]
    \label{def:sp_partial_obs}
  Let $\{F_i\}_{i = 1}^k$ be $k$ cumulative distribution functions with support 
  $[0, 1]$, i.e. $F_i(x) = 0\ \forall x < 0$ and $F_i(1) = 1$. A sample $(r, Y, Z)$ from a
  first-price auction with bid distributions $\{F_i\}_{i = 1}^k$ is generated as
  follows:
  \begin{enumerate}
    \item we, the observer, pick a price $r \in [0, 1]$, and let $X_{k + 1} = r$
    \item generate $X_i \ts F_i$ independently for all $i \in [k]$,
    \item observe a winner $Z = \argmax_{i \in [k+1]} X_i$ and an indicator $Q$
    indicating whether the reserve price $r$ was triggered.
  \end{enumerate}
\end{definition}

We again operate in the {\em effective-support setting} (cf. 
\cref{thm:1stPrice:likely}), and---that is, we
a tuple $p, \gamma \in [0, 1]$ such that for all $j \in [k]$, 
$\prod_{l \neq j} F_l(p) \geq \gamma$. 
In other words, the transaction price of
the auction will be less than $p$ with probability at least $\gamma$.

It turns out that in this seemingly limited observation model, a very
simple algorithm suffices for recovering agents' value distributions.
We 
begin with the Lemma demonstrating pointwise recovery of the bid
distributions for any $x \in [p, 1]$:
\begin{lemma}
\label{lem:sp_reserve_conc}
Fix any $x \in [p, 1]$ and any $\epsilon > 0$. Using $n$ samples from the 
we can obtain as estimate $\widehat{F}_{j \in [k]}(x)$ satisfying,
for all $j \in [k]$,
\[
    \abs*{\widehat{F}_j(x) - F_j(x)} \leq \epsilon
    \quad 
    \text{ with probability at least } 1 - \delta, 
\]
as long as $n \geq \frac{48}{\gamma \epsilon^2}\log(2k/\delta)$.
\end{lemma}
\begin{proof}
First, suppose we set the reserve price of the auction to $x$, and define the
random variable  $Z_{j}$ as the indicator of whether either (a) 
agent $j$ won the auction {\em and} the reserve price was triggered; or 
(b) no one won the auction and the reserve price was triggered.
By construction (and since (a) and (b) are disjoint),
\begin{align*}
\mathbb{P}(Z_j = 1) &= \mathbb{P}(X_j \geq x, X_{-j} \leq x)
                                +  \mathbb{P}(X_{[k]} \leq x) 
                    =  (1 - F_j(x)) \prod_{l \neq j} F_l(x)
                                   + \prod_{l \in [k]} F_l(x)
                    = \prod_{l \neq j} F_l(x).
\end{align*}
Thus, applying a (multiplicative) Chernoff bound combined with the lower bound 
$\prod_{l \neq j} F_l(x) \geq \gamma$ given by our effective support assumption,
\[
    \mathbb{P}
    \lr{
        \abs*{\sum_{i=1}^n Z_j^{(i)} - \prod_{l \neq j} F_l(x)} 
        \geq \epsilon \cdot \prod_{l \neq j} F_l(x)
    } \leq 2\exp\lbrb{-\epsilon^2 \gamma n/3}
\]
Now, by our effective support assumption, as long as $k \geq 2$, \\
\[
    \widehat{F}_j(x) \coloneqq 
    \frac{\prod_{l \in [k]}\lr{\sum_{i=1}^n Z_j^{(i)}}^{\frac{1}{k-1}}}
    { \lr{\sum_{i=1}^n Z_j^{(i)}} }
    \leq \frac{(1 + \epsilon)^{k/(k-1)} \prod_{l \in [k]} F_l(x)}
              {(1 - \epsilon) \prod_{l \neq j} F_l(x)} 
    \leq (1 + 4\epsilon) F_j(x) 
\]
and an identical argument for the lower bound shows that 
$\abs{\widehat{F}_j(x) - F_j(x)} \leq 4\epsilon$. Applying a union bound over
all agents completes the proof.
\end{proof}

We can use this result to construct piecewise-constant approximations of
$F_j(x)$ that is $\epsilon$-close to the true bid distributions:
\begin{theorem}
    \label{thm:sp_bid_insert}
    Assume the partially observed second-price setting, and suppose the cumulative
    density functions $F_{[k]}$ are all Lipschitz-continuous with Lipschitz
    constant $L$. For any pair $p, \gamma \in [0, 1]$ that define an effective
    support,
    we can find piecewise-constant functions $\widehat{F}_j(\cdot)$ satisfying
    \[
        \sup_{x \in [p, 1]} \abs*{\widehat{F}_j(x) - F_j(x)} \leq \epsilon
        \quad 
        \text{ with probability at least $1 - \delta$,}
    \] 
    using \(
        n = \Theta\lr{\frac{k \log(k/\epsilon) \log(L/\epsilon)^2}{\epsilon^3\gamma}}
    \)
    samples from the partially observed second-price model.
\end{theorem}
Given \cref{lem:sp_reserve_conc}, we can use the exact binary search and
estimation procedure from \cref{sec:1stPrice:additionalBidder} to
prove \cref{thm:sp_bid_insert}---we give the full proof in
\cref{app:sp_partial_info_proof}.

\section{Conclusion}
In this work, we presented efficient methods for estimating first- and
second-price auctions under independent (asymmetric) private values and partial
observability. 
Our methods come with convergence guarantees that are uniform in that
their error rates do not depend on the bid/value distributions being estimated.
These methods and the corresponding finite-sample guarantees build on a long
line of work in Econometrics that establishes either identification
results, or estimation results under restrictive assumptions such as
symmetry or full bid observability.

\section{Acknowledgements}
This work is supported by NSF Awards CCF-1901292, DMS-2022448 and
DMS2134108, a Simons Investigator Award, the Simons Collaboration on the Theory
of Algorithmic Fairness, a DSTA grant, the DOE PhILMs project
(DE-AC05-76RL01830), an Open Philanthropy AI Fellowship and a Microsoft Research-BAIR Open Research Commons grant.

\printbibliography

\appendix

\section{Omitted Proofs for First-Price Auctions}
\subsection{Omitted Calculations from Proof of \cref{thm:fp_partial_obs}}
\label{sec:fp_partial_info_extras}
\newcommand{\ygap}{\delta}
\newcommand{\numsteps}{T}
\newcommand{\numquantiles}{N}
\newcommand{\Hquants}[1]{v_{{#1}}}
\newcommand{\estHquants}[1]{\hat{v}_{{#1}}}
\newcommand{\Hiquants}[2]{w_{{#2}}}
\newcommand{\estHiquants}[2]{\hat{w}_{{#2}}}
\newcommand{\yquants}[1]{u_{{#1}}}
\newcommand{\hoeffgap}{\beta}
\newcommand{\binsearchgap}{\epsilon_1}
\newcommand{\binHstepgap}{\epsilon_1/2}
\newcommand{\binHistepgap}{\epsilon_1/2}
We further condition on the above and define an estimate of $G_i(x_j) 
= \int_{x_j}^1 \frac{1}{H(z)}\, dH_i(z)$,
\[
    \hat{G}_i(x_j) = \sum_{s=j}^{|X|-1} \lr{\hat{H}_i(x_{s+1}) - \hat{H}_i(x_s)}/{\hat{H}(x_s)}.
\]
Using the mean value theorem, there exists a set of points $\{\zeta_j\}$ with
$\zeta_j \in (x_{j}, x_{j+1}]$ and 
\[
    G_i(x_j) 
    = \int_{x_j}^1 \frac{1}{H(z)}\, dH_i(z)
    = \sum_{s=j}^{|X|-1} \frac{H_i(x_{s+1}) - H_i(x_{s})}{H(\zeta_s)}.
\]
We first bound the difference between our piecewise estimate and the true $G_i$
on the set $X$: 
\begin{align*}
    \abs{G_i(\hat{v}_t) - \hat{G}_i(\hat{v}_t)} 
    &=  
    \abs*{
    \sum_{s=j}^{|X|-1} \frac{H_i(x_{s+1}) - H_i(x_{s})}{H(\zeta_s)} - 
    \sum_{s=j}^{|X|-1} \frac{\hat{H}_i(x_{s+1}) - \hat{H}_i(x_s)}{\hat{H}(x_s)}
    } \\
    &\leq 
    \abs*{
        \sum_{s=j}^{|X|-1} \frac{H_i(x_{s+1}) - H_i(x_{s})}{H(\zeta_s)} - 
        \sum_{s=j}^{|X|-1} \frac{\hat{H}_i(x_{s+1}) - \hat{H}_i(x_s)}{H(\zeta_s)}
    } \\
    &\qquad +
    \abs*{
        \sum_{s=j}^{|X|-1} \frac{\hat{H}_i(x_{s+1}) - \hat{H}_i(x_s)}{H(\zeta_s)}
        - \sum_{s=j}^{|X|-1} \frac{\hat{H}_i(x_{s+1}) - \hat{H}_i(x_s)}{\hat{H}(x_s)}
    } \\
    &\leq 
    \frac{2}{\gamma} \cdot \sum_{s=1}^{|X|} \abs{H_i(x_{s}) - \hat{H}_i(x_{s})}
    + \max_{s \in [|X|]}\ \ \abs*{
        \frac{1}{H(x_{s+1})} - \frac{1}{\hat{H}(x_s)}
    } \\
    &\leq 
    \frac{2|X|\hoeffgap}{\gamma} + \frac{\ygap + 2 \cdot \binsearchgap + \hoeffgap}{\gamma^2}
    \\
    &\leq 
    \frac{2|X|\hoeffgap}{\gamma} + \frac{1}{\gamma^2} \max_{s \in [|X|]}\ \ \abs*{
        {H(\zeta_s)} - \hat{H}(x_s)
    } 
\end{align*}
\subsection{Proof of \cref{lem:qi_lb}}
We will proceed similarly to the proof of \citep{lebrun2006uniqueness}, who
use a similar technique to prove strict monotonicity (i.e., a lower bound of
zero). In particular, proving a quantitative lower bound requires carefully 
controlling additional terms that cancel in the original proof.

For $1 \leq i \leq n$, we define 
\[
    b'_i = \inf \left\{
        b' \in [0, 1]:\ 
        \frac{d}{db}\log\lr{G_i(v_i(b))} > L(b)\ \text{ for all } b \in (b', 1]
    \right\},
\]
and let $i$ be such that $b'_i = \max_{1 \leq k \leq n} b'_k$. 
Our goal is to prove that $b_i' < \rho$, since (by construction) this would
imply that our desired property is true on the entire range.

By continuity of $(d/db) \log\, G_i(v_i(b))$ and of $L(b)$, at the point
$b_i'$ we must have that
\[
    \qi = L(b).
\] 
Suppose that $b_i' \geq \rho$---by our definition of effective
support, $\Gi{b_i'} \geq \gamma$. Re-arranging the characterization of the
Bayes-Nash equilibrium \eqref{eq:equilibrium_characterization} (cf.
\cref{lemma:characterization_equilibrium}), 
\begin{align*}
    (v_i(b) - b) \cdot \frac{d}{db} \log\lr{G_i(v_i(b))} = \frac{1}{n-1}\lr{-(n-2) + \sum_{j \neq i} \frac{v_i(b) - b}{v_j(b) - b}}.
\end{align*}
Taking the derivative with respect to $b$ yields
\begin{align}
    D(b) 
    \label{eq:lebrun_derivative_1}
    &= \sum_{j \neq i} \frac{v_i'(b)}{v_j(b) - b} - \sum_{j \neq i} \frac{(v_i(b) - b)v_j'(b)}{(v_j(b) - b)^2} + \sum_{j \neq i} \frac{v_i(b) - v_j(b)}{(v_j(b) - b)^2}. 
\end{align}

\noindent Our next goal is to upper-bound the value of
\eqref{eq:lebrun_derivative_1} at $b_i'$. 
First, note that for all $j \neq i$, our construction of $b_i'$ implies that
$b_i' \in [b_j', 1]$ (since $i = \arg\max_k b_k'$), and so
\begin{equation*}
    \frac{1}{v_i(b_i') - b_i'} - \frac{1}{v_j(b_i') - b_i'} = \frac{d}{db} \log\lr{G_j(v_j(b_i'))} - \frac{d}{db} \log\lr{G_i(v_i(b_i'))} \geq 0,
\end{equation*}
meaning that 
\begin{equation}
    \vi{b_i'} \leq \vj{b_i'} \implies 
    \frac{v_i(b) - v_j(b)}{(v_j(b) - b)^2} \leq 0.
    \label{eq:value_monotone_1}
\end{equation}
Thus, we can safely ignore the (negative) final term when upper bounding
\eqref{eq:lebrun_derivative_1}. 
Turning our attention to the second term, \eqref{eq:log_deriv_diff}
implies the existence of at least one $j \neq i$ such that  
\begin{align}
    \nonumber
    \qj &\geq \frac{1}{n-1} \frac{1}{v_i(b) - b}, \text{ and rearranging, } \\
    \label{eq:value_monotone_2}
    (v_i(b) - b) v_j'(b_i') &\geq \frac{G_j(v_j(b_i'))}{(n-1)\cdot g_j(v_j(b_i'))} \geq \frac{\gamma}{(n-1)\cdot \eta}.
\end{align}
Since $v_j(b), b \in [0, 1]$ and $v_j(b_i') \geq v_i(b_i')$, we have $0 \leq v_j(b_i') - b_i' \leq 1$,
and in turn 
\begin{equation}
    \frac{(v_i(b) - b) v_j'(b_i')}{(v_j(b_i') - b_i')^2} \geq \frac{\gamma}{(n-1)\cdot \eta}
\end{equation}
for at least one $j \neq i$, and thus 
\begin{equation}
    -\sum_{j \neq i} \frac{(v_i(b) - b) v_j'(b_i')}{(v_j(b_i') - b_i')^2} \leq -\frac{\gamma}{(n-1)\cdot \eta}.
\end{equation}
Finally, recall that 
\[
    L(b_i') = \frac{v_i'(b_i') \cdot g_i(\vi{b_i'})}{G_i(\vi{b_i'})}
    \implies 
    v_i'(b_i') = \frac{L(b_i') \cdot G_i(\vi{b_i'})}{g_i(\vi{b_i'})}.
\]
Combining this with \eqref{eq:value_monotone_1} and
\eqref{eq:value_monotone_2}, and using that $\vj{b_i'} \geq \vi{b_i'}$ yields
\begin{align*}
    D(b_i') &\leq \frac{-\gamma}{(n-1)\cdot \eta} + \sum_{j \neq i} \frac{L(b_i')}{v_j(b_i') - b_i'} 
    \leq \frac{-\gamma}{(n-1)\cdot \eta} + \frac{(n - 1) \cdot L(b_i')}{\vi{b_i'} - b_i'}
    = 0,
\end{align*}
where the last equality is by definition of $L$. Meanwhile, by definition
\begin{align}
    \nonumber
    D(b) &= \frac{d}{db} \lbrb{ (v_i(b) - b) \cdot \frac{d}{db} \log\lr{G_i(v_i(b))} } \\ 
    \nonumber
    &= \lr{v'_i(b) - 1} \cdot \qi 
    + (v_i(b) - b) \cdot \frac{d^2}{db^2} \log\, G_i(v_i(b)) \\
    \label{eq:Db_lhs}
    D(b_i') &\geq (0 - 1)\cdot L(b_i') + (v_i(b_i') - b_i') \cdot \frac{d^2}{db^2} \log\, G_i(v_i(b)).
\end{align}
Combining the above results yields
\(
    (v_i(b_i') - b_i') \cdot \frac{d^2}{db^2} \log\, G_i(v_i(b))
    \leq -L(b_i') < 0,
\)
and since $v_i(b_i') - b_i' > 0$,
this must mean $\qi < 0$. However, this in turn implies that there exists an
$\epsilon > 0$ such that $\log\, \Gi{b_i'} < x$, which is a contradiction of
our definition of $\bi'$. Thus, our initial assumption ($b_i' > \delta$) is
impossible, which proves the desired result.
\section{Omitted Proofs for Second-Price Auctions}
\label{sec:proof_sec_price}

\subsection{Miscellaneous Results}
\label{ssec:misc_res_sec_price}

Here, we present miscellaneous results used in various parts of our proof. The first lemma shows that the functions, $U^*_i$, for different $i$ are within a constant factor of each other.

\uiujcomp*
\begin{proof}
    We have:
    \begin{equation*}
        U^*_i (x) - U^*_i (y) = \int_{y}^x \sum_{l \neq i} f_l (z) \prod_{m \neq i, l} F_m (z) dz \leq \lprp{\frac{\eta}{\alpha}}^3 (U^*_j (x) - U^*_j (y))
    \end{equation*}
    where the second inequality follows from \cref{as:second_price_bdd}:
    \begin{equation*}
        \forall l, l', z \in [0, 1]: f_l(z) \prod_{m \neq i, l} F_m (z) \leq \lprp{\frac{\eta}{\alpha}}^3 \cdot f_{l'}(z) \prod_{m \neq j, l'} F_{m} (z).
    \end{equation*}
\end{proof}

\begin{lemma}
    \label{lem:uh_disc_apx}
    We have, for all $i \in [k]$ and all $\tau \in [T]$,
    \begin{align*}
        \forall l \in \ell^{(\tau)}: \lprp{\hat{U}_i (x_{\tau, l}) - \hat{U}_i (x_{\tau, l - 1})} \leq \frac{1}{1 - x_{\tau, l}} \cdot \Delta^{(\tau)}_{i, l} \leq 2 \lprp{\hat{U}_i (x_{\tau, l}) - \hat{U}_i (x_{\tau, l - 1})} \\
        \forall \hat{U}_i (x_\tau) - \hat{U}_i (x_{\tau - 1}) \leq \sum_{l = 1}^{\ell^{(\tau)}} \frac{1}{1 - x_{\tau, l}} \cdot \Delta^{(\tau)}_{i, l} \leq 2 \cdot \lprp{\hat{U}_i (x_\tau) - \hat{U}_i (x_{\tau - 1})}
    \end{align*}
\end{lemma}
\begin{proof}
    For the lower bound, we have $\forall l \in \ell^{(\tau)}$:
    \begin{align*}
        \frac{1}{1 - x_{\tau, l}} \cdot \Delta^{(\tau)}_{i, l} &= \frac{1}{n} \cdot \sum_{j = 1}^n \frac{1}{(1 - x_{\tau, l})} \cdot \bm{1} \lbrb{Z_j = i, x_{\tau, l - 1} < Y_j \leq x_{\tau, l}} \\
        &\geq \frac{1}{n} \cdot \sum_{j = 1}^n \frac{1}{(1 - Y_j)} \cdot \bm{1} \lbrb{Z_j = i, x_{\tau, l - 1} < Y_j \leq x_{\tau, l}} = \hat{U}_i (x_{\tau, l}) - \hat{U}_i (x_{\tau, l - 1}).
    \end{align*}
    Summing the above inequality over $l$ concludes the proof of the lower bound. Similarly, for the upper bound, we get $\forall l \in \ell^{(\tau)}$:
    \begin{align*}
        &\frac{1}{1 - x_{\tau, l}} \cdot \Delta^{(\tau)}_{i, l} - (\hat{U}_i (x_{\tau, l}) - \hat{U}_i (x_{\tau, l - 1})) \\
        &= \frac{1}{n} \cdot \sum_{j = 1}^n \lprp{\frac{1}{(1 - x_{\tau, l})} - \frac{1}{(1 - Y_j)}} \cdot \bm{1} \lbrb{Z_j = i, x_{\tau, l - 1} < Y_j \leq x_{\tau, l}} \\
        &\leq \frac{1}{n} \cdot \sum_{j = 1}^n \lprp{\frac{\delta}{(1 - x_{\tau, l})(1 - Y_j)}} \cdot \bm{1} \lbrb{Z_j = i, x_{\tau, l - 1} < Y_j \leq x_{\tau, l}} \\
        &\leq \frac{1}{n} \cdot \sum_{j = 1}^n \lprp{\frac{1}{2 \cdot (1 - Y_j)}} \cdot \bm{1} \lbrb{Z_j = i, x_{\tau, l - 1} < Y_j \leq x_{\tau, l}} = \frac{1}{2} \lprp{\hat{U}_i (x_{\tau, l}) - \hat{U}_i (x_{\tau, l - 1})}.
    \end{align*}
    Again, re-arranging and summing over $l$ concludes the proof. 
\end{proof}

\subsection{Proof of \cref{thm:sp_bid_insert}}
\label{app:sp_partial_info_proof}

We combine this result with the same binary search from
\cref{sec:1stPrice:additionalBidder} to identify {\em quantiles} of the
functions $\widehat{F}_i$. In particular, fix $\epsilon > 0$ and let
\[
    W = \{w_a \coloneqq \gamma + a \cdot \frac{\epsilon}{2} 
    \ |\  a \in \mathbb{N} \text{ and } \gamma + a \cdot \frac{\epsilon}{2} \leq 1 \} 
    \cup \{1\}.
\]
An identical argument from the one from the proof of \cref{thm:fp_partial_obs}
shows that for any $\epsilon > 0$, we can use binary search to 
find $\hat{z}_{j,a}$ such that  
\begin{align}
    \label{eq:sp_quantile_recovery}
    \abs{F_j(\hat{z}_{j,a}) - w_a} \leq \frac{\epsilon}{2} 
    \text{ for all $j \in [k]$ and $a \in [|W|]$}
    \\ 
    \nonumber
    \text{ w.p. } \qquad 
    1 - \frac{4k}{\epsilon}\log\lr{\frac{4L}{\epsilon}}\exp\lbrb{-\epsilon^2 \gamma n_1 / 192}
\end{align}
using $C \cdot n_1\cdot k \cdot \log(4L/\epsilon) / \epsilon$ samples 
for a universal constant $C$ 
(in particular, see 
the proof of \cref{thm:fp_partial_obs}, set $\delta = \epsilon_1 =
\epsilon/2$, and use \cref{lem:sp_reserve_conc} for the pointwise guarantee in
place of \eqref{eq:fp_pointwise_conc}).

Conditioning on this event, we can thus define the piecewise-constant functions
$\widehat{F}_j$ for  $j \in [k]$ as \[
\widehat{F}_j(x) = \sum_{a \in [|W|]} \bm{1}\lbrb{x \in [z_{j,a}, z_{j, a+1})} \cdot \lr{\gamma + a\cdot \frac{\epsilon}{2}}
\]
Now, consider any $x \in [p, 1]$ and define $a \in \mathbb{N}$ such that 
$x \in [\hat{z}_{j, a}, \hat{z}_{j,a+1}]$, so that by construction 
$\widehat{F}_j(x) = w_{j,a}$. Then:
\begin{align*}
    F_j(x) &\geq F_j(\hat{z}_{j,a}) &\text{(monotonicity of CDF)}\\
            &= w_{j,a+1}
                + (F_j(\hat{z}_{j,a}) - w_{j,a}) 
                + (w_{j,a} - w_{j,a+1}) \\
            &\geq w_{j,a+1} - \epsilon/2 - \epsilon/2 &\text{(definition of $W$ and \eqref{eq:sp_quantile_recovery})} \\
            &\geq \widehat{F}_j(x) - \epsilon &\text{(definition of $\widehat{F}_j$)}
\end{align*}
Similarly, $F_j(x) \leq \widehat{F}_j(x) + \epsilon$. 

It remains to handle $x \in [p, \hat{z}_{j,0}]$: note that by (strict)
monotonicity of the $\widehat{F}_j$, we must have $\widehat{F}_j^{-1}(\gamma)
\leq p$. Thus, if $p \leq x \leq \hat{z}_{j,0}]$,
\begin{align*}
    F_i(x) \geq F_i(p) \geq \gamma = \widehat{F_i}(x),
\qquad \text{ and } \qquad 
    F_i(x) \leq F_i(\hat{z}_{j,0}) \leq \gamma + \frac{\epsilon}{2} \leq \widehat{F}_i(x) + \epsilon.
\end{align*}

Thus, 
$\abs{F_j(x) - \widehat{F}_j(x)} \leq \epsilon$ over the entire interval $[p,
1]$ with probability at least $1 - \delta$, as long as
\[
    n \geq \frac{C k \log(k/\epsilon) \log(L/\epsilon)^2}{\epsilon^3\gamma},
\]
for a universal constant $C$, completing the proof.
\qed

\end{document}